\documentclass[11pt,a4paper]{article}
\usepackage[margin=1in]{geometry}
\usepackage[T1]{fontenc}
\usepackage{lmodern}
\usepackage{amsmath,amssymb,amsthm,mathtools}
\usepackage{microtype}
\usepackage{needspace}
\usepackage[framemethod=tikz]{mdframed}
\usetikzlibrary{arrows.meta}
\mdfdefinestyle{resultstatement}{backgroundcolor=black!5,hidealllines=true,
 skipabove=10pt,skipbelow=10pt,innerleftmargin=5pt,innerrightmargin=5pt,
 innertopmargin=5pt,innerbottommargin=5pt,splittopskip=8pt,splitbottomskip=5pt,
 needspace=4\baselineskip,nobreak=false}
\newmdenv[style=resultstatement,hidealllines=false,
 linewidth=0.5pt,linecolor=black!45,innerleftmargin=8pt,innerrightmargin=8pt]{mainresult}
\usepackage[hidelinks]{hyperref}
\hypersetup{pdftitle={Logarithmic asymptotics for polynomial threshold functions},
 pdfauthor={Martin Anthony}}

\newtheorem{theorem}{Theorem}

\newtheorem{lemma}{Lemma}
\surroundwithmdframed[style=resultstatement]{lemma}

\newcommand{\F}{\mathbb F}
\newcommand{\R}{\mathbb R}
\newcommand{\RM}{\operatorname{RM}}
\newcommand{\rank}{\operatorname{rank}}
\newcommand{\Span}{\operatorname{span}}

\newcommand{\wt}{\operatorname{wt}}

\newcommand{\E}{\mathbb E}
\newcommand{\ind}{\mathbf 1}
\title{Logarithmic asymptotics for polynomial threshold functions}
\author{Martin Anthony\thanks{Department of Mathematics,
London School of Economics and Political Science.\newline
Email: \href{mailto:m.anthony@lse.ac.uk}{\texttt{m.anthony@lse.ac.uk}}.}}
\date{}

\begin{document}
\maketitle

\begin{abstract}
We count Boolean functions represented as the sign of a real polynomial
of bounded degree on the $n$-dimensional cube. We prove that the
logarithm of their number is asymptotic to that of the classical
hyperplane-arrangement upper bound, uniformly over all degrees, with
relative error $O((\log n)^2/n)$. We also obtain sharper additive and
relative error estimates in particular degree ranges, including an
additive error $O(n^2)$ between the logarithms throughout a small
linear degree range and improved bounds just below the central degree.
The proof combines probabilistic estimates for Reed--Muller codes
with a deletion argument for sets of evaluation vectors, an entropy
count, and a sharp-threshold estimate for expected rank deficiency.
\end{abstract}

\setcounter{tocdepth}{2}
{\small\tableofcontents}

\clearpage

\section{Introduction}

\subsection{Definitions and main results}

For integers $1\le d\le n$, let $V=\{0,1\}^n$ and
$N=|V|=2^n$. The space of real multilinear polynomials of degree at
most $d$ has dimension
\[
 m=m_{n,d}:=\sum_{j=0}^d\binom nj.
\]
Write $T_{n,d}$ for the number of functions $F:V\to\{-1,1\}$ of
the form $F(x)=\operatorname{sgn}p(x)$, where $p$ is a real
multilinear polynomial of degree at most $d$ with no zero values on
$V$. Multilinearity loses no functions, since reducing powers by
$x_i^2=x_i$ preserves cube values and does not increase degree.
We could equally well define these functions as the signs of real
multilinear polynomials of degree at most $d$, with no zero values,
on the domain $\{-1,1\}^n$. The substitution $x_i=(y_i+1)/2$
preserves multilinearity and degree and gives a bijection between the
functions on the two cubes, so the count is unchanged.

The equations $p(x)=0$, one for each $x\in V$, define a central
hyperplane arrangement in the $m$-dimensional real coefficient space.
Here central means that each hyperplane contains the origin.
Its chambers, the connected components of the complement, correspond
to the functions counted by $T_{n,d}$. See Cover~\cite[Sections~I--III]{Cover1965}
and Baldi and Vershynin~\cite[Lemmas~2.1 and~3.2]{BaldiVershynin2019}.
Indeed, the signs of the evaluations are constant on each chamber,
while prescribed strict output signs specify an intersection of open
half-spaces, which is convex and therefore connected, and hence is a
single chamber when non-empty. The classical
arrangement bound~\cite{Cover1965,Winder1966} therefore gives
\begin{equation}\label{eq:cover}
 T_{n,d}\le Q_{n,d}:=2\sum_{i=0}^{m-1}\binom{N-1}{i}.
\end{equation}
We consider whether the logarithm of this upper bound is attained
when the degree grows with the number of variables.

Define the relative logarithmic error by
\begin{equation}\label{eq:relative-error}
 E_{n,d}:=\frac{\log_2Q_{n,d}-\log_2T_{n,d}}{\log_2Q_{n,d}}
 =1-\frac{\log_2T_{n,d}}{\log_2Q_{n,d}}.
\end{equation}
We state relative error bounds throughout so that the results can
be compared directly. Theorems~\ref{thm:main} and~\ref{thm:away-centre}
also give additive bounds on the difference between the logarithms
in their respective degree ranges.
Our first result determines the leading logarithm for every degree,
with a uniform bound on this error.

\begin{mainresult}
\begin{theorem}[Logarithmic asymptotics for every degree]\label{thm:all-degrees}
There are absolute constants $C>0$ and $n_0$ such that
\[
 0\le E_{n,d}\le C\frac{(\ln n)^2}{n}
 \qquad(n\ge n_0,\ 1\le d\le n).
\]
In particular, $\log_2T_{n,d}\sim\log_2Q_{n,d}$ uniformly over
$1\le d\le n$.
\end{theorem}
\end{mainresult}

Baldi and Vershynin~\cite[Conjecture~8.2]{BaldiVershynin2019},
citing earlier work of the author, state the conjecture that
$T_{n,d}\sim Q_{n,d}$ for every degree. They suggest that this
assertion might be too strong, but that the corresponding logarithmic
equivalence might hold. Theorem~\ref{thm:all-degrees} establishes
this logarithmic form for every degree, with a uniform rate.

For $d\ge\lceil n/2\rceil$, the author's earlier centrality
theorem~\cite[Theorem~1]{Anthony2026} gives the stronger equivalence
\[
 T_{n,d}\sim Q_{n,d}\sim2^N.
\]
Theorem~\ref{thm:all-degrees} is proved here without using that
result. Theorems~\ref{thm:main}--\ref{thm:near-central} give
estimates for $d<n/2$, and Theorem~\ref{thm:near-central}(i),
together with monotonicity in the degree, also establishes the
upper-range bound in Theorem~\ref{thm:all-degrees}.
At $d=n$, the equality $T_{n,n}=Q_{n,n}=2^N$ is exact.

The lower bounds on $T_{n,d}$ use vectors associated with the cube
vertices. For each vertex $x$, the values at $x$ of all monomials of
degree at most $d$ form its \emph{evaluation vector}.
Section~\ref{sec:evaluation} explains the connection with the
hyperplane arrangement described above and shows that counting
certain sets of these vectors, called \emph{good sets}, gives a
lower bound on $T_{n,d}$. Good sets satisfy a condition stronger than
linear independence. We use three methods to count them, each suited
to a different degree range. For small degrees, we bound the
probability of linear dependence to count good sets directly. For
intermediate degrees, we first count independent sets and then delete
a small number of vectors to make them good. Near the central degree,
we instead estimate how many vectors must be removed from a random
sample to leave an independent set. An entropy argument turns a small
expected loss into a count of independent sets, to which the same
deletion argument applies.

Our second result gives a quantitative comparison throughout a small
linear degree range. It is the main estimate used to prove
Theorem~\ref{thm:all-degrees} for small degrees.

\begin{mainresult}
\begin{theorem}[Additive and relative error for small degrees]\label{thm:main}
There are absolute constants $C>0$ and $n_0$ such that, whenever
$n\ge n_0$ and $2\le d\le10^{-4}n$,
\[
 0\le\log_2Q_{n,d}-\log_2T_{n,d}
 \le(n+1)\log_2(N/m_{n,d})+C,
 \qquad E_{n,d}\le C\frac{n}{m_{n,d}}.
\]
Consequently, the additive error is $O(n^2)$, and
\[
 \log_2T_{n,d}\sim\log_2Q_{n,d}
 \sim m_{n,d}\log_2(N/m_{n,d})
\]
uniformly throughout this degree range as $n\to\infty$.
\end{theorem}
\end{mainresult}

The constant $10^{-4}$ is a convenient choice for the estimates in the
proof (see the remarks at the end of Subsection~\ref{sec:small-degree-completion}).
Since $\log_2(N/m)=n-\log_2m$, this factor is asymptotic to $n$
precisely when $\log_2m/n\to0$. When $d$ is proportional to $n$,
$m$ grows exponentially with $n$, so $\log_2m$ is of order $n$
and cannot be discarded. For sublinear degree, however, the bound
$m\le(en/d)^d$ gives
\[
 \frac{\log_2m}{n}
 \le\frac dn\log_2\!\left(\frac{en}{d}\right)\longrightarrow0.
\]
Thus Theorem~\ref{thm:main}, together with Zuev's degree-one result,
implies, for every integer sequence $1\le d=o(n)$,
\begin{equation}\label{eq:sublinear}
 \log_2T_{n,d}\sim n m_{n,d}=n\sum_{j=0}^d\binom nj.
\end{equation}

Our next result gives additive and relative error bounds over
a wider range below $n/2$.

\begin{mainresult}
\begin{theorem}[Additive and relative error below the central degree]\label{thm:away-centre}
There are absolute constants $C>0$ and $n_0$ such that, whenever
$n\ge n_0$, $n^{-1/3}\le\varepsilon<1/2$ and
$2\le d\le(1/2-\varepsilon)n$,
\[
 0\le\log_2Q_{n,d}-\log_2T_{n,d}
 \le Cm_{n,d}\ln(e/\varepsilon),
 \qquad
 E_{n,d}\le\frac{C\ln(e/\varepsilon)}{n\varepsilon^2}.
\]
\end{theorem}
\end{mainresult}

\newpage

Our fourth result gives a further estimate as the degree approaches
the central degree from below. Put
\[
 \alpha=\alpha_{n,d}=\frac{m_{n,d}}{2^n},\qquad
 a=\ln(1/\alpha),
\]
and write $H_2(u)=-u\log_2u-(1-u)\log_2(1-u)$ for binary
entropy, with $0\log_2 0=0$. For $d<n/2$, the condition
$a=o(n)$ is equivalent to $n/2-d=o(n)$.

\begin{mainresult}
\begin{theorem}[Relative error just below the central degree]\label{thm:near-central}
Recall that $a=\ln(2^n/m_{n,d})$. There are absolute constants
$\delta_0,\delta_1,C>0$ and $n_0$ such that, whenever $n\ge n_0$
and $1\le d<n/2$, the following hold.

\textup{(i)} If $a\le\delta_0n$, then
\begin{equation}\label{eq:near-central-direct}
 0\le E_{n,d}\le C\frac{a+\ln(n/a)}n.
\end{equation}

\textup{(ii)} If $(a/n)\ln(a+2)\le\delta_1$, then
\begin{equation}\label{eq:degree-dependent-error}
 0\le E_{n,d}\le\frac{C\{1+\ln(a+2)\ln(n/a)\}}n.
\end{equation}

In particular, if $d<n/2$ and $n/2-d=o(n)$, then
\[
 \log_2T_{n,d}\sim\log_2Q_{n,d}\sim NH_2(\alpha).
\]
\end{theorem}
\end{mainresult}

Part~(i) applies throughout
\[
 \frac n2-c_0n\le d<\frac n2
\]
for a sufficiently small absolute constant $c_0>0$. Part~(ii)
applies throughout
\[
 \frac n2-\frac{c_1n}{\sqrt{\ln n}}\le d<\frac n2
\]
for any fixed $0<c_1<\sqrt{\delta_1/2}$. Both assertions hold for
all sufficiently large $n$, uniformly over the indicated degrees.
Subsection~\ref{sec:central-degree-ranges} justifies these ranges and
shows where the respective hypotheses fail.

Theorem~\ref{thm:main} gives relative error $O(n/m_{n,d})$
for small degrees. Theorem~\ref{thm:away-centre} reaches larger
degrees, with relative error $O(1/n)$ when $d$ stays a fixed
proportion below $n/2$. Comparing with Theorem~\ref{thm:away-centre}
at fixed $\varepsilon$, Theorem~\ref{thm:main}'s bound has smaller
order for $d\ge3$ in their common range, while at $d=2$ the bounds
have the same order. Theorem~\ref{thm:near-central}
gives sharper estimates closer to $n/2$ and covers degrees beyond
the range of Theorem~\ref{thm:away-centre}. Where the ranges
overlap, we use whichever bound is smaller.
Theorem~\ref{thm:all-degrees} gives a single bound valid for every
degree, while Theorems~\ref{thm:main}--\ref{thm:near-central}
give smaller errors in particular ranges.

\subsection{Related work and proof strategy}

For degree one, Zuev~\cite{Zuev1989} established
$\log_2T_{n,1}\sim n^2$. Kahn, Koml\'os and
Szemer\'edi~\cite[Section~4]{KahnKomlosSzemeredi1995} obtained
$\log_2T_{n,1}=n^2-n\log_2n+O(n)$.
Comparison with the same expansion for $\log_2Q_{n,1}$ gives
$\log_2Q_{n,1}-\log_2T_{n,1}=O(n)$, which we use in the
degree-one case of Theorem~\ref{thm:all-degrees}.
More recently, Xie~\cite[Theorem~1]{Xie2026} has obtained
\[
 T_{n,1}=2\binom{2^n-1}{n}\bigl(1+O(n^{-99})\bigr).
\]
Since the last summand dominates the sum defining $Q_{n,1}$,
this gives $T_{n,1}\sim Q_{n,1}$.

For higher degrees, the classical arrangement bound gives
$T_{n,d}\le Q_{n,d}$. See Baldi~\cite{Baldi1988} and
Anthony~\cite{Anthony1995} for its application to polynomial
threshold functions. A recursive argument due to
Saks~\cite{Saks1993} gives
\[
 \log_2T_{n,d}\ge\binom{n}{d+1}\qquad(1\le d<n).
\]
For each fixed $d\ge2$, the lower bound has leading term
$n^{d+1}/(d+1)!$, whereas $\log_2Q_{n,d}\sim n^{d+1}/d!$.
These bounds therefore determine the order of the logarithm,
but leave a factor of $d+1$ between their leading terms.
An account of these earlier bounds is given in
Anthony~\cite[Sections~4.5--4.7]{Anthony2001}.
Baldi and Vershynin~\cite[Theorem~1.1 and the subsequent remarks]{BaldiVershynin2019}
closed this gap for every fixed degree. More generally, their
quantitative bounds apply for $d\le n^{1-\varepsilon}$, for each
fixed $0<\varepsilon<1$, and yield
$\log_2T_{n,d}\sim nm_{n,d}$ when $d=o(\log n)$.

Anthony~\cite[p.~96]{Anthony1995} observed that a lower bound for
$T_{n,d}$ asymptotically matching the classical arrangement upper
bound $Q_{n,d}$ would imply the upper-degree assertion of the
Wang--Williams conjecture~\cite{WangWilliams1991}.

The proofs combine probabilistic estimates with counting arguments
that turn those estimates into lower bounds for $T_{n,d}$.
The counting methods have antecedents in work of
Odlyzko~\cite{Odlyzko1988}, Zuev~\cite{Zuev1992} and
Irmatov~\cite{Irmatov2001,Irmatov2018}, and in Baldi and
Vershynin's treatment of polynomial threshold
functions~\cite{BaldiVershynin2019}.
The connection with coding theory comes from Reed--Muller codes,
which record the values of polynomials over $\mathbb F_2$ on the
Boolean cube. Recovery asks whether missing values can be
determined from those that are known.
Baldi and Vershynin use results of Abbe, Shpilka and
Wigderson~\cite{AbbeEtAl2015}. Luh~\cite{Luh2026} obtains
related estimates for real linear independence at fixed degree.

For Theorem~\ref{thm:main}, we sharpen a first-moment argument of
Bhandari et al.~\cite{BhandariEtAl2022}, based on estimates of
Sberlo and Shpilka~\cite{SberloShpilka2020}. The direct count
requires strong bounds on the probability of linear dependence,
which we establish in the small-degree range. This gives a direct
count of good sets with the precision needed for the additive
$O(n^2)$ error bound.

For Theorem~\ref{thm:away-centre}, we use a result of Sberlo and
Shpilka~\cite[Theorem~1.2]{SberloShpilka2020} to count linearly
independent sets of evaluation vectors. The deletion argument from
Section~\ref{sec:evaluation} then converts this into a count of
good sets. For small degrees, the required bound follows from
Theorem~\ref{thm:main}.

Near the central degree, the sampling method used for
Theorem~\ref{thm:away-centre} would require more columns than are
available. For Theorem~\ref{thm:near-central}, we instead sample
about $m$ columns and control the expected nullity. We combine a
sharp-threshold result of O'Donnell~\cite[equation~(10.8)]{ODonnell2014}
with an entropy argument to obtain this count, then apply the same
deletion argument to count good sets. The approach is related to
work of Kudekar et al.~\cite{KudekarEtAl2017} on Reed--Muller codes
and to the author's centrality paper~\cite{Anthony2026}.

Sections~\ref{sec:binary-estimates} and~\ref{sec:all-degrees}
prove Theorems~\ref{thm:main}, \ref{thm:away-centre}, and
\ref{thm:near-central}(i). The supporting weight and sharp-threshold
estimates are proved in Appendices~\ref{app:weight-estimates} and
\ref{app:recovery-estimates}. Appendix~\ref{app:near-central-ii}
completes the proof of Theorem~\ref{thm:near-central}(ii).

The deduction of Theorem~\ref{thm:all-degrees} is given in
Section~\ref{sec:all-degree-deduction}, after the main counting arguments. Degree one follows from
Kahn, Koml\'os and Szemer\'edi. For a suitable
$\varepsilon_n$ of order $(\ln n)^{-1/2}$, Theorem~\ref{thm:away-centre}
covers $2\le d\le(1/2-\varepsilon_n)n$, and
Theorem~\ref{thm:near-central}(ii) covers the remaining degrees below
$n/2$. Applying part~(i) just below the centre and using monotonicity
in $d$ gives the upper range. Thus the all-degree rate is obtained
by combining the estimates where their ranges overlap.

Figure~\ref{fig:proof-dependencies} summarises the three counting
routes and shows how the resulting estimates are used to prove
Theorem~\ref{thm:all-degrees}.

\begin{figure}[htbp]
\centering
\begin{tikzpicture}[
 stage/.style={draw=black!45,fill=black!4,rounded corners=2pt,
   text width=3.9cm,minimum height=1.85cm,align=center,
   font=\footnotesize,inner sep=5pt},
 common/.style={draw=black!50,fill=black!9,rounded corners=2pt,
   text width=13.4cm,minimum height=0.65cm,align=center,
   font=\footnotesize,inner sep=5pt},
 flow/.style={-{Stealth[length=2mm]},line width=0.6pt,shorten <=1pt,shorten >=1pt}]
 \node[common] at (0,0.9)
   {Common setup: evaluation columns over \mbox{$\F_2$} and \mbox{$\R$}
    (Section~\ref{sec:evaluation})};
 \node[font=\small\bfseries] at (-4.6,0) {Theorem~\ref{thm:main}};
 \node[stage] (B1) at (-4.6,-1.2) {First moment\\(Lemma~\ref{lem:conditioning})\\Bias estimates\\(Lemmas~\ref{lem:weight-estimates} and~\ref{lem:large-bias})};
 \node[stage] (B2) at (-4.6,-3.55) {Bound on \mbox{$q_r$}\\\eqref{eq:rank-bound}};
 \draw[flow] (B1) -- (B2);
 \node[stage] (B3) at (-4.6,-5.9) {Good sets directly\\\eqref{eq:good-probability}};
 \draw[flow] (B2) -- (B3);
 \node[stage] (B4) at (-4.6,-8.25) {Lower bound for \mbox{$T_{n,d}$}\\\eqref{eq:enumeration-transfer}};
 \draw[flow] (B3) -- (B4);
 \node[font=\small\bfseries] at (0,0) {Theorem~\ref{thm:away-centre}};
 \node[stage] (C1) at (0,-1.2) {Full-rank sampling\\(Lemma~\ref{lem:rank-retention})};
 \node[stage] (C2) at (0,-3.55) {Count independent sets\\\mbox{$J_m$}};
 \draw[flow] (C1) -- (C2);
 \node[stage] (C3) at (0,-5.9) {Count good sets\\by deletion\\(Lemma~\ref{lem:central-deletion})};
 \draw[flow] (C2) -- (C3);
 \node[stage] (C4) at (0,-8.25) {Lower bound for \mbox{$T_{n,d}$}\\\eqref{eq:central-deletion-count}};
 \draw[flow] (C3) -- (C4);
 \node[font=\small\bfseries] at (4.6,0) {Theorem~\ref{thm:near-central}};
 \node[stage] (D1) at (4.6,-1.2) {Integral identities and\\sharp threshold\\(Lemmas~\ref{lem:recovery-integrals} and~\ref{lem:recovery-bounds})};
 \node[stage] (D2) at (4.6,-3.55) {Expected nullity\\(Lemmas~\ref{lem:central-nullity} and~\ref{lem:retention-excess})};
 \draw[flow] (D1) -- (D2);
 \node[stage] (D3) at (4.6,-5.9) {Count independent sets\\from expected nullity\\(Lemma~\ref{lem:entropy-independent})};
 \draw[flow] (D2) -- (D3);
 \node[stage] (D4) at (4.6,-8.25) {Lower bound for \mbox{$T_{n,d}$}\\by deletion (Lemma~\ref{lem:central-deletion}),\\giving \eqref{eq:entropy-transfer}};
 \draw[flow] (D3) -- (D4);
 \node[common] (comparison) at (0,-10.2)
   {Lower bounds for \mbox{$\log_2T_{n,d}$}, compared with
    \mbox{$\log_2Q_{n,d}$}\\(Lemma~\ref{lem:cover-estimates})};
 \draw[flow] (B4.south) -- (B4.south |- comparison.north);
 \draw[flow] (C4.south) -- (comparison.north);
 \draw[flow] (D4.south) -- (D4.south |- comparison.north);

 \node[font=\footnotesize,align=center,text width=13.4cm] at (0,-11.4)
   {\textbf{Application to the proof of Theorem~\ref{thm:all-degrees}}
    (Subsection~\ref{sec:all-degree-proof})\\[2pt]
    \mbox{$\varepsilon_n=c_1/\sqrt{\ln n}$}, with fixed
    \mbox{$0<c_1<\sqrt{\delta_1/2}$}, \mbox{$n$} sufficiently large};
 \node[stage,minimum height=1.2cm,anchor=south] (proofB) at (-4.6,-14.85)
   {\textbf{Theorem~\ref{thm:main}}\\[3pt]
    \mbox{$2\le d\le10^{-4}n$}};
 \node[stage,minimum height=1.2cm,anchor=south] (proofC) at (0,-14.85)
   {\textbf{Theorem~\ref{thm:away-centre}}\\[3pt]
    \mbox{$2\le d\le(1/2-\varepsilon_n)n$}};
 \node[stage,minimum height=2.4cm,anchor=south] (proofD) at (4.6,-14.85)
   {\textbf{Theorem~\ref{thm:near-central}}\\[3pt]
    Part~(ii):\\
    \mbox{$(1/2-\varepsilon_n)n<d<n/2$}\\[3pt]
    Part~(i) at \mbox{$d_0=\lfloor n/2\rfloor-1$},\\
    then monotonicity for\\\mbox{$d\ge\lceil n/2\rceil$}};
 \draw[flow] (proofB.east) -- (proofC.west);
 \node[common] (theoremA) at (0,-15.8)
   {\textbf{Theorem~\ref{thm:all-degrees}:} \mbox{$1\le d\le n$}.
    The classical degree-one estimate supplies \mbox{$d=1$}};
 \draw[flow] (proofC.south) -- (theoremA.north);
 \draw[flow] (proofD.south) -- (proofD.south |- theoremA.north);
\end{tikzpicture}
\caption{The counting routes and their use in proving Theorem~\ref{thm:all-degrees}.
Arrows point from an input to the step or result that uses it.
Parenthetical numerical references are to equations.
The lower panel shows only the degree ranges used in that proof.
The theorem statements give their full hypotheses.
Theorem~\ref{thm:main} enters through the small-degree case of
Theorem~\ref{thm:away-centre}.}
\label{fig:proof-dependencies}
\end{figure}
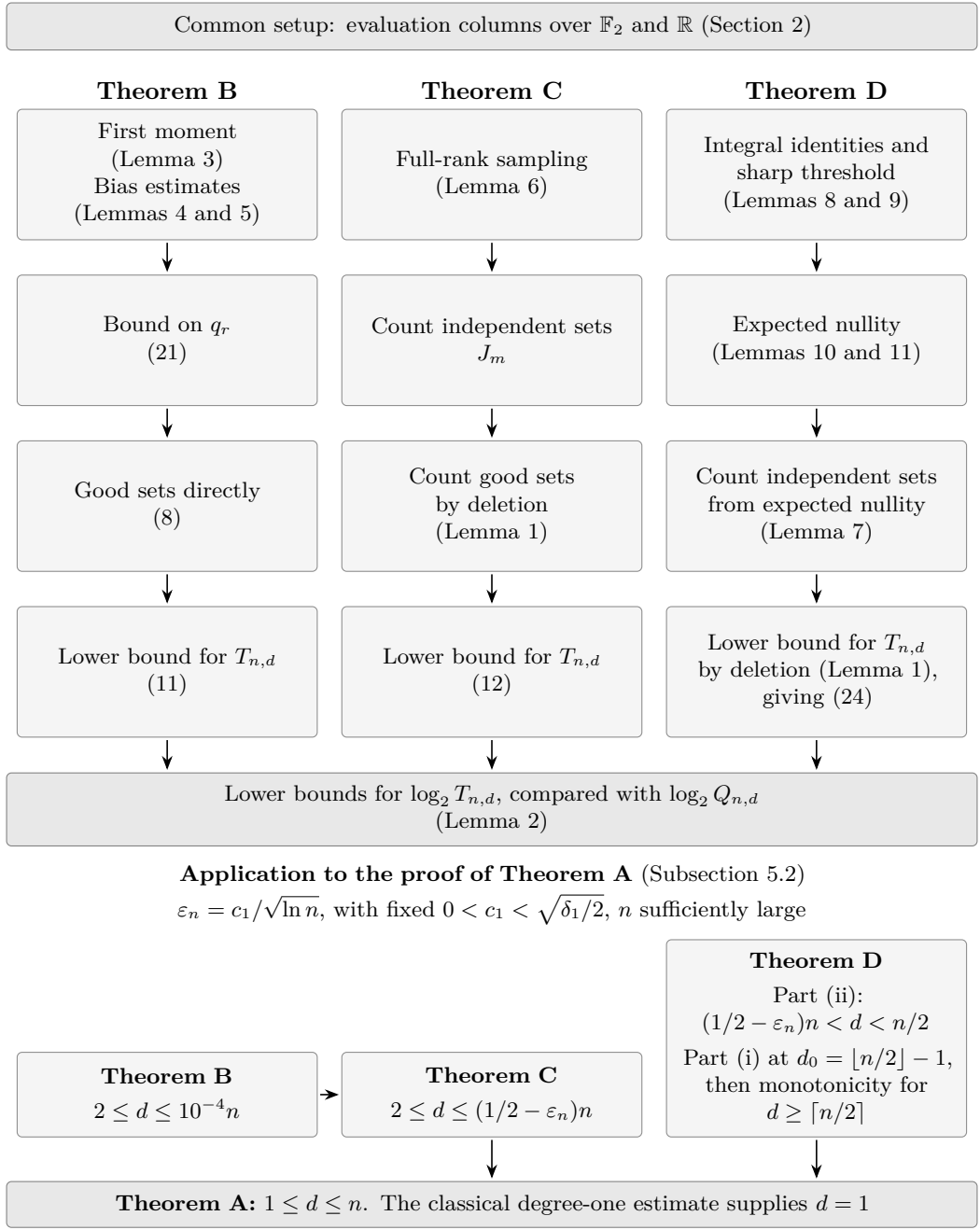
\clearpage

The proofs repeatedly use the following standard binomial entropy
bounds (see Cover and Thomas~\cite[Example~11.1.3]{CoverThomas2006}):
\begin{equation}\label{eq:binomial-entropy}
 \frac{2^{MH_2(k/M)}}{M+1}
 \le\binom Mk
 \le\sum_{j=0}^k\binom Mj
 \le2^{MH_2(k/M)},
 \qquad 0\le k\le M/2,
\end{equation}
where $M\ge1$ and $k$ are integers.

\section{Evaluation vectors and counting}
\label{sec:evaluation}

\subsection{The evaluation matrix over two fields}

We use the standard polynomial evaluation matrix (see Baldi and
Vershynin~\cite[Section~3]{BaldiVershynin2019}).
Recall that $N=2^n$. For $S\subseteq[n]=\{1,2,\ldots,n\}$,
write $x_S=\prod_{i\in S}x_i$, with $x_\varnothing=1$.
Let $A$ be the $m\times N$ matrix whose rows are indexed by the
sets $S\subseteq[n]$ with $|S|\le d$ and whose columns are indexed
by the vertices of $V=\{0,1\}^n$. Its column at $x\in V$ is
\[
 A_x=(x_S)_{|S|\le d}.
\]
A real coefficient vector $b=(b_S)_{|S|\le d}$ determines the
polynomial
\[
 p_b(x)=\sum_{|S|\le d}b_Sx_S=b^{\mathrm{T}}A_x.
\]
Thus $A_x$ records the values of all the monomials at one vertex,
whereas $A^{\mathrm{T}}b$ gives the values of one polynomial at all
vertices. The entries of $A$ are all equal to $0$ or $1$. We can
therefore also regard $A$ as a matrix over the field
$\F_2=\{0,1\}$, in which addition and multiplication are taken
modulo $2$. For a binary coefficient vector $c\in\F_2^m$, the same
matrix gives the binary polynomial function
\[
 f_c(x)=\sum_{|S|\le d}c_Sx_S\pmod2,
 \qquad (f_c(x))_{x\in V}=A^{\mathrm{T}}c\quad\hbox{over }\F_2.
\]
Here the values of $f_c$ are in $\F_2$. The real polynomial $p_b$
and the binary polynomial $f_c$ belong to different spaces. What
connects them is the common integer matrix $A$. Whenever a rank or
span is taken below, we specify the field.

For $U\subseteq V$, write $A_U$ for the submatrix consisting of the
columns indexed by $U$. We will need the following elementary comparison, also noted by
Baldi and Vershynin~\cite[Section~4.2]{BaldiVershynin2019}:
\begin{equation}\label{eq:rank-fields}
 \rank_{\F_2}A_U\le\rank_{\R}A_U.
\end{equation}
Indeed, binary rank at least $r$ means that some $r\times r$ minor
has non-zero determinant modulo $2$. Its integer determinant is odd
and hence non-zero over $\R$. This proves \eqref{eq:rank-fields}.

The full matrix $A$ has row rank $m$ over either field. To see this,
restrict to the $m$ columns $x=\mathbf1_T$ with $|T|\le d$, where
$\mathbf1_T$ is the indicator vector of $T$. The entry in row $S$ and the column corresponding to $\mathbf1_T$
is $1$ precisely when $S\subseteq T$. Using the same ordering of
the subsets by increasing cardinality for the rows $S$ and columns
$\mathbf1_T$ gives an upper triangular matrix with diagonal entries $1$. In particular, different coefficient
vectors give different polynomial functions on $V$, over either field.

\subsection{Sets that determine their spans}

We obtain a lower bound for $T_{n,d}$ by counting distinct
intersections of the hyperplanes. This suffices because the number
of chambers is at least the number of such intersections.
We follow the construction of Baldi and
Vershynin~\cite[Section~4.4]{BaldiVershynin2019}, using their
terminology of good sets.
Let $q_r$ be the probability that the columns of $A$ indexed by a
uniformly chosen $r$-element subset of $V$ are linearly dependent
over $\F_2$. Call a set $U\subseteq V$ \emph{good} if the columns
of $A_U$ are linearly independent over $\F_2$ and no column indexed
outside $U$ lies in their span over $\F_2$. For $1\le k<m$,
a $k$-element set $U$ is good exactly when the columns of
$A_{U\cup\{x\}}$ are linearly independent over $\F_2$ for every
$x\notin U$. For a uniformly chosen $k$-element
set $U$, let $Z$ count the vertices $x\notin U$ for which
$A_{U\cup\{x\}}$ has dependent columns over $\F_2$.
Then $U$ is good precisely when $Z=0$, so the first-moment bound gives
\[
 \Pr(U\text{ is not good})=\Pr(Z\ge1)\le\E Z.
\]
Choose $x$ uniformly from the $N-k$ vertices outside $U$.
Every $(k+1)$-element set arises from exactly $k+1$ equally likely
pairs $(U,x)$, so $U\cup\{x\}$ is uniformly distributed among
these sets. Consequently, $\E Z=(N-k)q_{k+1}$, and
\begin{equation}\label{eq:good-probability}
 \Pr(U\text{ is good})\ge1-(N-k)q_{k+1}.
\end{equation}

For any good $U$, put $k=|U|$. The columns of $A_U$ are independent over both
$\F_2$ and $\R$. Moreover, for each $x\notin U$,
\[
 k+1=\rank_{\F_2}A_{U\cup\{x\}}
 \le\rank_{\R}A_{U\cup\{x\}}\le k+1.
\]
Adding any column indexed outside $U$ increases the real rank from
$k$ to $k+1$. The real span of the columns of $A_U$ therefore
contains the evaluation columns indexed by $U$ and no other
evaluation columns. For $x\in V$, let $H_x=A_x^\perp$ be the hyperplane in $\R^m$
with normal vector $A_x$.
Different good sets give different real spans and
therefore different intersections
\[
 \bigcap_{x\in U}H_x
 =\bigl(\Span_{\R}\{A_x:x\in U\}\bigr)^\perp,
\]
since taking the orthogonal complement twice recovers the original
subspace.

We need a lower bound on chambers from these intersections. For a
finite central arrangement $\mathcal H$, let $r(\mathcal H)$ be
its number of chambers and $i(\mathcal H)$ the number of distinct
intersections of its hyperplanes. Include the ambient space as the
intersection of the empty family. Then
\begin{equation}\label{eq:regions-flats}
 r(\mathcal H)\ge i(\mathcal H).
\end{equation}
This inequality is due to Zuev~\cite[Section~3]{Zuev1992} (see also
Baldi and Vershynin~\cite[Lemma~2.2(2)]{BaldiVershynin2019}).

We sketch Zuev's induction on the number of hyperplanes. The empty
arrangement has one chamber and one intersection. Assume the
inequality holds for every arrangement of at most $k$ hyperplanes,
in every dimension. Let $\mathcal H$ have $k+1$ hyperplanes,
choose $H\in\mathcal H$, put $\mathcal H'=\mathcal H\setminus\{H\}$,
and let $\mathcal K$ be the arrangement in $H$ formed by the
distinct hyperplanes $H\cap L$ with $L\in\mathcal H'$.
Its chambers and intersections are taken within $H$, with $H$
itself as the empty intersection. A chamber of $\mathcal H'$
that meets $H$ is cut by $H$ into two chambers of $\mathcal H$,
and its intersection with $H$ is a chamber of $\mathcal K$.
This gives a bijection between such chambers and the chambers
of $\mathcal K$, so $r(\mathcal H)=r(\mathcal H')+r(\mathcal K)$.
Every intersection of $\mathcal H$ that is not an intersection
of $\mathcal H'$ involves $H$, and is therefore an intersection
of members of $\mathcal K$. Hence
$i(\mathcal H)\le i(\mathcal H')+i(\mathcal K)$.
Both $\mathcal H'$ and $\mathcal K$ have at most $k$ hyperplanes,
so induction gives
\[
r(\mathcal H)=r(\mathcal H')+r(\mathcal K)
\ge i(\mathcal H')+i(\mathcal K)\ge i(\mathcal H).
\]

Different good sets give distinct intersections of the hyperplanes.
Hence
\begin{equation}\label{eq:good-sets-count}
 T_{n,d}\ge\#\{U\subseteq V:U\text{ is good}\}.
\end{equation}
In particular, restricting to sets of size $k$ and using
\eqref{eq:good-probability} gives
\begin{equation}\label{eq:enumeration-transfer}
 T_{n,d}\ge\bigl(1-(N-k)q_{k+1}\bigr)\binom Nk.
\end{equation}
For Theorem~\ref{thm:main}, it therefore remains to bound $q_{k+1}$
for suitable $k$ close to $m$. Because the error is multiplied by
$N-k$, this argument needs dependence probabilities on the scale
$1/N$. The next section obtains these estimates using binary polynomials.
For the other degree ranges, we first count independent sets and
then apply the following deletion argument.

\subsection{From independent sets to good sets}

For Theorems~\ref{thm:away-centre} and~\ref{thm:near-central},
we will first count independent sets of evaluation columns.
An independent set need not be good: its span may contain further
columns of the matrix. We show that deleting a controlled number
of columns produces a good set, and bound how many independent
sets can produce the same good set. Thus a lower bound on the
number of independent sets gives a lower bound on the number of
good sets and, by \eqref{eq:good-sets-count}, on $T_{n,d}$.

The evaluation words of binary polynomials of degree at most $d$ form the Reed--Muller code
$\RM(n,d)$, the row space of $A$ over $\F_2$.
For $d<n$, duality and the minimum-weight formula
\cite[Section~4.1]{Carlet2021} give
\[
 \ker_{\F_2}A=\RM(n,n-d-1),
\]
and every non-zero word of this dual code has at least $2^{d+1}$
non-zero coordinates. Thus every dependent set of evaluation columns
has at least $2^{d+1}$ elements.

The following deletion argument applies to any binary matrix with
distinct non-zero columns. As above, a good set has independent columns and its span contains
no further column of the matrix.

Each extra column in the span of an independent set has a unique expression as a sum of columns from that set. If every dependent set has at least $g$ elements, this expression uses at least $g-1$ columns. Deleting any one of them removes the extra column from the span. The proof chooses deletions at random and shows that a suitable choice breaks all these representations simultaneously.

\begin{lemma}[Deleting columns to obtain a good set]
\label{lem:central-deletion}
Suppose a binary matrix has $N$ distinct non-zero columns and every
set of columns dependent over $\F_2$ has at least $g\ge2$ elements. Let $J_r$ count its
independent $r$-element sets, where $1\le r\le N$, and put
\[
 \ell=\min\left\{r,\left\lceil\frac{r\ln(N+1)}{g-1}\right\rceil\right\}.
\]
Then the number of good sets of size $r-\ell$ is at least
\[
 \frac{J_r}{\binom{N-r+\ell}{\ell}}.
\]
\end{lemma}

\begin{proof}
Fix an independent set $S$ of size $r$. Each column $A_x$ outside $S$
that lies in its span has a unique expression
\[
 A_x=\sum_{u\in F_x}A_u,\qquad F_x\subseteq S.
\]
The columns indexed by $F_x\cup\{x\}$ are linearly dependent,
so the hypothesis gives $|F_x|+1\ge g$, and hence $|F_x|\ge g-1$.
There are at most $N-r\le N$ such extra columns $A_x$, and hence
at most $N$ sets $F_x$.

We will construct a set $H\subseteq S$ of at most $\ell$ columns
that meets every $F_x$. Deleting these columns will break the
representation of each extra column $A_x$. We will then check
that the remaining set is good.
If $\ell=r$, take $H=S$. Otherwise, choose $\ell$ elements
of $S$ independently and uniformly, allowing repetitions.
For each $F_x$, the probability that none of its members is
chosen is at most
\[
 (1-|F_x|/r)^\ell\le e^{-\ell(g-1)/r}\le\frac1{N+1}.
\]
By the union bound, the probability that at least one set $F_x$
is missed is at most
\[
 \frac{N}{N+1}<1.
\]
Thus some choice meets every $F_x$. Taking $H$ to be the set of
distinct columns chosen gives $|H|\le\ell$.
Enlarge $H$ to size $\ell$ if necessary.
The set $S\setminus H$ is good: independence excludes the deleted
columns from its span. Columns outside the span of $S$ remain outside. Because $H$
meets every $F_x$, the unique representation of each extra
column $A_x$ uses a deleted column. Hence none of these extra
columns lies in the span of $S\setminus H$. The empty set is also good, since no column
is zero.

Choose one such good subset from each independent $S$. A fixed good
set of size $r-\ell$ is contained in at most
$\binom{N-r+\ell}{\ell}$ sets of size $r$. Counting these preimages
proves the bound.
\end{proof}

For the evaluation matrix with $d<n$, write $J_r$ for the number of
independent $r$-element sets of columns over $\F_2$, and take
$g=2^{d+1}$. The geometric reduction
\eqref{eq:good-sets-count} and Lemma~\ref{lem:central-deletion} give,
whenever $J_r>0$,
\begin{equation}\label{eq:central-deletion-count}
 \log_2T_{n,d}\ge\log_2J_r-\ell\log_2N
 \ge\log_2J_r-O\left(\frac{rn^2}{2^d}+n\right).
\end{equation}

\subsection{Estimates for the arrangement upper bound}
\label{sec:cover-estimates}

The three counting arguments produce different lower bounds for
$T_{n,d}$. We collect the estimates for $Q_{n,d}$ with which they
will be compared. Part~(i) retains the precision needed for the
additive bound in Theorem~\ref{thm:main}. Parts~(ii) and~(iii)
give the leading logarithm in the other degree ranges.

\begin{lemma}\label{lem:cover-estimates}
Let $N=2^n$ and $m=m_{n,d}$, where $1\le d\le n$.
\begin{enumerate}
\item[\textup{(i)}] If $m^2/N\to0$, then
\[
 Q_{n,d}=2(1+o(1))\binom N{m-1}.
\]
\item[\textup{(ii)}] If $m\le N/2$, then
\[
 \log_2Q_{n,d}=m\log_2(N/m)+O(m+n),
 \qquad
 \log_2Q_{n,d}\ge1+(m-1)\log_2(N/m).
\]
\item[\textup{(iii)}] If $1\le d<n/2$, put
$\alpha=m/N$, $a=\ln(1/\alpha)$, and $\Lambda=NH_2(\alpha)$.
Then
\[
 \frac{ma}{\ln2}\le\Lambda\le\frac{m(a+1)}{\ln2},
 \qquad \log_2Q_{n,d}\le\Lambda+1,
 \qquad \log_2Q_{n,d}=\Lambda+O(n).
\]
For all sufficiently large $n$, uniformly over these degrees,
$\log_2Q_{n,d}\ge\Lambda/2$.
\end{enumerate}
\end{lemma}

\begin{proof}
For part~(i), the binomial coefficients in $Q_{n,d}$ increase up
to the last term for large $n$, and
\[
 \sum_{i=0}^{m-2}\binom{N-1}{i}
 \le(m-1)\binom{N-1}{m-2}
 =\frac{(m-1)^2}{N-m+1}\binom{N-1}{m-1}
 =o\!\left(\binom{N-1}{m-1}\right).
\]
Also,
$\binom{N-1}{m-1}=(1-(m-1)/N)\binom N{m-1}$,
which proves the assertion.

For part~(ii), monotonicity of the summands and
$\binom{N-1}{m-1}=(m/N)\binom Nm$ give
\[
 2\frac mN\binom Nm\le Q_{n,d}
 \le2m\frac mN\binom Nm.
\]
Taking logarithms and using
$(N/m)^m\le\binom Nm\le(eN/m)^m$ proves the first estimate.
For the lower bound, use
$Q_{n,d}\ge2\binom{N-1}{m-1}\ge2(N/m)^{m-1}$.

For part~(iii), the definition of entropy and
$0\le-(1-\alpha)\ln(1-\alpha)\le\alpha$ give the bounds on
$\Lambda$. The binomial entropy estimates
\eqref{eq:binomial-entropy} give
\[
 \frac{2m}{N(N+1)}2^\Lambda
 \le2\frac mN\binom Nm\le Q_{n,d}
 \le2\sum_{i=0}^{m}\binom Ni\le2^{\Lambda+1}.
\]
These imply both estimates for $\log_2Q_{n,d}$.
Finally, $m\ge n+1$ and monotonicity of $H_2$ on $[0,1/2]$ give
\[
 \Lambda\ge NH_2\!\left(\frac{n+1}{N}\right)
 \ge(n+1)\{n-\log_2(n+1)\}.
\]
This grows at least quadratically in $n$, so the $O(n)$ error is
at most $\Lambda/2$ for all sufficiently large $n$, uniformly in $d$.
\end{proof}

The constants in the $O$-terms are absolute, and the estimate in
\textup{(i)} is uniform whenever $m^2/N\to0$ uniformly.

\section{Counts for small degrees}
\label{sec:binary-estimates}

In this section we prove Theorem~\ref{thm:main}. We first bound
the probability that a random set of evaluation columns is linearly
dependent over $\F_2$. We obtain this bound using estimates for the
weights of Reed--Muller codewords. Finally, we apply the geometric
reduction in Section~\ref{sec:evaluation} to count polynomial
threshold functions.

\subsection{Codewords, weights and bias}

Recall that the set of binary evaluation vectors
\[
 C:=\{A^{\mathrm{T}}c:c\in\F_2^m\}\subseteq\F_2^N
\]
is the binary Reed--Muller code $\RM(n,d)$ (see
Carlet~\cite[Section~4.1]{Carlet2021}). A \emph{codeword} is a vector in $\F_2^N$
obtained by evaluating one of these binary polynomials at all $N$ vertices. We identify a word $f\in C$ with the function
$x\mapsto f(x)$ giving its entries. The set $C$ is a linear subspace:
adding two words corresponds to adding their coefficient vectors over
$\F_2$. Its dimension over $\F_2$ is $m$, by the rank calculation
above, so it contains exactly $2^m$ words. Since each codeword has a unique coefficient vector, choosing a word
uniformly from $C$ is equivalent to choosing its $m$ coefficients
independently and uniformly from $\F_2$.

For a word $f$, its \emph{relative Hamming weight} is
\[
 \wt(f):=\frac{|\{x\in V:f(x)=1\}|}{N}.
\]
Thus weight here means the fraction of non-zero entries, rather than
the number of non-zero entries. We also use its \emph{bias}
\[
 B_f:=\frac1N\sum_{x\in V}(-1)^{f(x)}=1-2\wt(f).
\]
The bias is the fraction of zero values minus the fraction of one
values. A balanced word has bias $0$, and the two constant words have
bias $1$ and $-1$. Since the constant polynomial $1$ belongs to
$C$, the map $f\mapsto f+1$ is a bijection of $C$ that exchanges
zeros and ones. It takes $B_f$ to $-B_f$, so the bias of a uniformly chosen
word has a symmetric distribution.
For a uniformly chosen codeword $f$, write $X=|B_f|$ for its
absolute bias.

The coding estimates we use count words of small weight. To express
them in terms of bias, put
\[
 W_{n,d}(u):=|\{f\in C:\wt(f)\le u\}|.
\]
For $0<b\le1$, the event $B_f\ge b$ is precisely
$\wt(f)\le(1-b)/2$. Symmetry gives the same number of words with
$B_f\le-b$, and these two events are disjoint. Consequently
\begin{equation}\label{eq:weight-bias}
 \Pr_f(X\ge b)=\frac{2W_{n,d}((1-b)/2)}{2^m}.
\end{equation}
The subscript $f$ indicates that the probability is over a uniformly
chosen word of $C$. The magnitude $X$ has at most $N+1$ possible
values, since the number of ones in a word is an integer between $0$
and $N$.

The minimum non-zero relative weight of $\RM(n,d)$ is $2^{-d}$
(see Carlet~\cite[Section~4.1]{Carlet2021}). Consequently,
\begin{equation}\label{eq:minimum-weight}
 W_{n,d}(u)=1\qquad\hbox{if }0\le u<2^{-d},
\end{equation}
since only the zero word has weight below this threshold.

\subsection{The first moment and repeated vertices}
\label{sec:independence}

We begin with the first-moment calculation of
Bhandari et al.~\cite[Section~1.1]{BhandariEtAl2022}, expressed in
our notation.
We then refine this calculation to bound dependence for sets much
closer to full rank.
Let $x_1,\ldots,x_r$ be independent uniformly chosen vertices of $V$,
and let $\rho$ be the rank over $\F_2$ of the matrix with columns
$A_{x_1},\ldots,A_{x_r}$.
The kernel of this matrix has $2^{r-\rho}$ elements, each specifying
a subset of the sample positions whose columns sum to zero. Thus
$2^{r-\rho}-1$ counts the non-empty such subsets. It is zero when
the columns are independent and at least $1$ otherwise, so its
expectation bounds the probability of dependence. Equivalently,
this is the union bound over non-empty zero-sum subsets. We
calculate the expectation by counting the codewords that vanish
at the sampled vertices.

Vanishing at all sampled vertices means that the coefficient vector
lies in the kernel of the transpose of this matrix. By rank--nullity,
this kernel has dimension $m-\rho$, giving exactly $2^{m-\rho}$ such
words.
Choose a codeword $f$ uniformly from $C$, independently of the sampled
vertices. We calculate the probability that $f$ vanishes at all sampled
vertices in two ways. For a fixed sample, this probability is
$2^{m-\rho}/2^m=2^{-\rho}$. For a fixed word $f$, a uniformly chosen
vertex is a zero of $f$ with probability $1-\wt(f)=(1+B_f)/2$.
Since the vertices are chosen independently, the probability that $f$
vanishes at all of them is $((1+B_f)/2)^r$. Consequently,
\[
 \E_{\mathrm{sample}}2^{-\rho}
 =\E_f\left[\left(\frac{1+B_f}{2}\right)^r\right],
\]
where the expectation on the left is over the sampled vertices and
that on the right is over the uniformly chosen codeword.
Multiplying by $2^r$ gives
\begin{equation}\label{eq:vanishing-count}
 \E_{\mathrm{sample}}2^{r-\rho}=\E_f(1+B_f)^r.
\end{equation}

\begin{lemma}\label{lem:conditioning}
Suppose $d\ge1$, $1\le r\le m$, and $M:=\binom r2<N$.
For $0<\delta<1$ with $r\delta\le1$, put
\[
 L_r=\E_f\bigl[(1+X)^r\ind_{\{\delta\le X<1\}}\bigr],
 \qquad X=|B_f|.
\]
Then
\[
 q_r\le
 \frac{2^{r-m}+r^4\delta^2/N+\binom M2/N^2+L_r}{1-M/N}.
\]
\end{lemma}

\begin{proof}
For the sample $x_1,\ldots,x_r$ drawn with replacement, define the
random variable $Y=2^{r-\rho}$.
By \eqref{eq:vanishing-count},
\[
\E_{\mathrm{sample}}Y=\E_f(1+B_f)^r
=\frac1{2^m}\sum_{f\in C}(1+B_f)^r.
\]
Each codeword is chosen with probability $2^{-m}$. We bound this expectation
by considering separately the codewords with $|B_f|<\delta$,
$\delta\le|B_f|<1$, and $|B_f|=1$. For
$|b|<\delta$, the binomial expansion gives
\[
 \frac{(1+b)^r+(1-b)^r}{2}
 \le1+Mb^2+r^4b^4.
\]
For the remainder, use $r|b|\le1$ and
$\sum_{j\ge2}1/(2j)!<1$.
We next bound the contribution to $\E_f(1+B_f)^r$ from codewords
with $|B_f|<\delta$. Complementing a codeword changes $B_f$ to
$-B_f$ and preserves $|B_f|$. Hence, by the preceding inequality,
\[
\begin{aligned}
\E_f\bigl[(1+B_f)^r\ind_{\{|B_f|<\delta\}}\bigr]
&=\E_f\left[\frac{(1+B_f)^r+(1-B_f)^r}{2}
             \ind_{\{|B_f|<\delta\}}\right]\\
&\le\E_f\bigl[(1+MB_f^2+r^4B_f^4)
               \ind_{\{|B_f|<\delta\}}\bigr]\\
&\le1+(M+r^4\delta^2)\E_f B_f^2\\
&=1+\frac MN+\frac{r^4\delta^2}{N}.
\end{aligned}
\]
For the second inequality, we use $B_f^4\le\delta^2B_f^2$
when $|B_f|<\delta$, then drop the indicator from the resulting
non-negative expression. The last equality uses
\[
 \E_f B_f^2
 =\frac1{N^2}\sum_{x,y\in V}\E_f(-1)^{f(x)+f(y)}
 =\frac1N
\]
(see also Luh~\cite[Section~6.2, p.~11]{Luh2026}). Indeed, the
$N$ diagonal terms are $1$. If $x\ne y$, choose a coordinate
$i$ on which they differ. Since $d\ge1$, the map $f\mapsto f+x_i$
is a bijection of $C$ and reverses the sign of
$(-1)^{f(x)+f(y)}$. Every off-diagonal expectation is therefore
zero.
In $\E_f(1+B_f)^r$, the zero word has $B_f=1$ and contributes
$2^{-m}2^r=2^{r-m}$. The constant-one word has $B_f=-1$ and
contributes zero. For codewords with $\delta\le X<1$, we have
$(1+B_f)^r\le(1+X)^r$, so their total contribution to the
expectation is at most $L_r$. Combining these estimates and
subtracting $1$ gives
\begin{equation}\label{eq:iid-first-moment}
 \E Y-1\le2^{r-m}+\frac MN+\frac{r^4\delta^2}{N}+L_r.
\end{equation}

We now pass from sampling with replacement to a uniformly chosen
set of $r$ distinct vertices. Samples containing repetitions have
dependent columns, and we will subtract a lower bound for their
contribution to $\E(Y-1)$. For $1\le i<j\le r$, let
$E_{ij}=\{x_i=x_j\}$. This is a pair-collision event, meaning
that the same vertex is selected at positions $i$ and $j$.
Let $E=\bigcup_{i<j}E_{ij}$, the event that the sample contains
a repeated vertex. There are $M=\binom r2$ such pair-collision
events, each of probability $N^{-1}$. The
intersection of any two distinct such events has probability
$N^{-2}$, whether or not the pairs share an index. The first two Bonferroni inequalities therefore give
\[
 \frac MN-\frac{\binom M2}{N^2}\le\Pr(E)\le\frac MN.
\]
On $E$ the columns are dependent, so $Y-1\ge1$. Conditional on the complementary event $E^c$, in which all sampled
vertices are distinct, the unordered sample is a uniformly chosen
$r$-element set.
Recall that $q_r$ is the probability that the evaluation columns
indexed by a uniformly chosen $r$-element subset of $V$ are
linearly dependent over $\F_2$. Since $q_r$ is the conditional probability of dependence given
$E^c$, $q_r\Pr(E^c)$ is the probability that the sample has
distinct vertices and linearly dependent evaluation columns. On this event,
$(Y-1)\ind_{E^c}\ge1$, and this random variable is non-negative
elsewhere. Therefore
\[
\begin{split}
q_r\Pr(E^c)
&\le\E\bigl[(Y-1)\ind_{E^c}\bigr]\\
&=\E(Y-1)-\E\bigl[(Y-1)\ind_E\bigr]\\
&\le\E(Y-1)-\Pr(E)\\
&\le2^{r-m}+\frac{r^4\delta^2}{N}
                  +\frac{\binom M2}{N^2}+L_r.
\end{split}
\]
The equality subtracts the contribution from samples containing
repetitions. This contribution is at least $\Pr(E)$ because
$Y-1\ge1$ on $E$.
Dividing by $\Pr(E^c)\ge1-M/N$ proves the result.
\end{proof}

\subsection{The contribution of larger biases}
\label{sec:large-bias}

We now bound $L_r$. The proof adapts the argument of Bhandari
et al.~\cite[Claim~3.6 and its proof]{BhandariEtAl2022}, using
the weight estimates of Sberlo and
Shpilka~\cite{SberloShpilka2020}.
Throughout this subsection, let $c=10^{-4}$, assume $2\le d\le cn$,
and write $W=W_{n,d}$. Put
\[
\gamma=d/n,\qquad L=\lceil3\log_2m\rceil.
\]
A larger bias magnitude $X$ gives a larger factor $(1+X)^r$
in the expectation defining $L_r$. It is not enough merely to
show that such words are rare. Their rarity must compensate for
this factor so that the expectation is small.
We first state the estimates needed to bound $L_r$.

\begin{lemma}\label{lem:weight-estimates}
For all sufficiently large $n$,
\begin{equation}\label{eq:low-weight-input}
\log_2W(2^{-j})\le\frac14m2^{-j}\qquad(j\ge2).
\end{equation}
Moreover,
\begin{equation}\label{eq:auxiliary-bias-bounds}
\begin{gathered}
 \log_2m<0.0015n,\qquad L<0.005n,\\
 e^{-24\gamma L}\ge2^{-d},
 \qquad \frac{m2^{-d}}n\longrightarrow\infty,
\end{gathered}
\end{equation}
and
\begin{equation}\label{eq:bias-tail}
\Pr_f(X\ge2^{-\ell})
\le2^{-(4/5)m e^{-24\gamma\ell}}
\qquad(1\le\ell\le L).
\end{equation}
The inequalities and the limit are uniform over $2\le d\le cn$.
\end{lemma}

The proof, given in Appendix~\ref{app:weight-estimates}, derives
these estimates from Corollaries~3.6 and~3.13 and Propositions~3.7
and~3.14 of Sberlo and Shpilka~\cite{SberloShpilka2020}.
The constant $24$ comes from the choice of parameters in their
bias estimate. 

\begin{lemma}\label{lem:large-bias}
Put $\delta=2^{-L}$, so that $\delta\le m^{-3}$.
For all sufficiently large $n$, uniformly over integers
$2\le d\le cn$ and $1\le r\le m$,
\[
L_r=\E_f\bigl[(1+X)^r\ind_{\{\delta\le X<1\}}\bigr]
\le N^{-2}.
\]
\end{lemma}

\begin{proof}
We use Lemma~\ref{lem:weight-estimates}. Put $s=m2^{-d}$.
We shall show that each possible value $x$ of $X$ with
$\delta\le x<1$ satisfies
$(1+x)^r\Pr_f(X=x)\le2^{-s/20}$.
There are at most $N+1$ such values, and $s/n\to\infty$ by
\eqref{eq:auxiliary-bias-bounds}, so summing these bounds will prove
the lemma.

First suppose $x<1/2$, and put
$\ell=\lceil\log_2(1/x)\rceil$. Then
$2\le\ell\le L$ and $2^{-\ell}\le x<2^{1-\ell}$. We have
\[
 \log_2(1+x)\le\frac x{\ln2}
 \le\frac{2^{1-\ell}}{\ln2}
 \le\frac34e^{-24\gamma\ell}.
\]
For the last inequality,
$2^{1-\ell}e^{24\gamma\ell}/\ln2$ decreases with $\ell$, and
at $\ell=2$ it is at most $e^{48c}/(2\ln2)<3/4$.
As $x\ge2^{-\ell}$, equation~\eqref{eq:bias-tail} bounds
$\Pr_f(X=x)$ by $2^{-(4/5)m e^{-24\gamma\ell}}$.
Since $\ell\le L$, equation~\eqref{eq:auxiliary-bias-bounds}
also gives $e^{-24\gamma\ell}\ge e^{-24\gamma L}\ge2^{-d}$.
Using $r\le m$, we obtain
\[
 \begin{aligned}
 (1+x)^r\Pr_f(X=x)
 &\le2^{(3/4)m e^{-24\gamma\ell}}
       2^{-(4/5)m e^{-24\gamma\ell}}\\
 &=2^{-m e^{-24\gamma\ell}/20}
 \le2^{-s/20}.
 \end{aligned}
\]

Now suppose $1/2\le x<1$, and put $u=(1-x)/2$. A codeword $f$ with
$|B_f|=x$ has relative weight $u$ or $1-u$, so either $f$ or
$f+1$ has relative weight $u$. Since $u>0$, the codeword of
weight $u$ is non-zero, and the minimum-weight bound
\eqref{eq:minimum-weight} gives $u\ge2^{-d}$.
Since $x\ge1/2$, we also have $u\le1/4$.
Choose $j=\lfloor\log_2(1/u)\rfloor$, so that
$2\le j\le d$ and $u\le2^{-j}<2u$.
Since $W$ is non-decreasing, equation~\eqref{eq:low-weight-input}
gives $W(u)\le W(2^{-j})\le2^{m2^{-j}/4}\le2^{mu/2}$.
By \eqref{eq:weight-bias},
\[
\Pr_f(X=x)\le\Pr_f(X\ge x)=2^{1-m}W(u).
\]
Since $r\le m$, $1+x=2(1-u)$ and
$\log_2(1-u)\le-u$, we obtain
\[
 \begin{split}
 (1+x)^r\Pr_f(X=x)
 &\le[2(1-u)]^m2^{1-m}W(u)\\
 &=2(1-u)^mW(u)\\
 &\le2^{1+m(\log_2(1-u)+u/2)}\\
 &\le2^{1-mu/2}\le2^{1-s/2}\le2^{-s/20}
 \end{split}
\]
for all sufficiently large $n$.

Finally, $s/n\to\infty$ uniformly by
\eqref{eq:auxiliary-bias-bounds}, so $s/20\ge3n+1$ for all
sufficiently large $n$. Summing over at most $N+1$ values of $X$
and using $N=2^n$ gives
\begin{equation}\label{eq:weighted-bias-total}
 \E_f\bigl[(1+X)^r\ind_{\{\delta\le X<1\}}\bigr]
 \le(N+1)2^{-s/20}\le2^{n+1-(3n+1)}=N^{-2}.
\end{equation}
\end{proof}

\subsection{Completion of the proof of Theorem~\ref{thm:main}}
\label{sec:small-degree-completion}

We now combine the dependence estimates to count good sets of
columns. By \eqref{eq:enumeration-transfer}, it suffices to choose
$k$ close to $m$ for which $(N-k)q_{k+1}$ stays bounded away from $1$.
A fixed positive proportion of the $\binom Nk$ sets will then
contribute distinct hyperplane intersections.

Throughout this subsection $2\le d\le10^{-4}n$, and all estimates
are uniform over these degrees. Let $1\le r\le m$. Recall that
$q_r$ is the probability that the evaluation columns indexed by
a uniformly chosen $r$-element subset of $V$ are linearly dependent
over $\F_2$.
We apply Lemma~\ref{lem:conditioning} with $M=\binom r2$ and
$\delta=2^{-L}$, where $L=\lceil3\log_2m\rceil$.
Then $\delta\le m^{-3}$ and $r\delta\le m^{-2}\le1$.
Lemma~\ref{lem:large-bias} gives $L_r\le N^{-2}$.
Since $r\le m$ and $M\le r^2/2$,
\[
 \frac{r^4\delta^2}{N}\le\frac1{m^2N},\qquad
 \binom M2\le\frac{m^4}8,\qquad
 1-\frac MN\ge1-\frac{m^2}{2N}.
\]
The bound $\log_2m<0.0015n$ in
\eqref{eq:auxiliary-bias-bounds} gives
$m^2/(2N)<2^{-0.997n-1}\to0$, so $M<N$ and the denominator
is positive for all sufficiently large $n$. Substitution in
Lemma~\ref{lem:conditioning} gives
\begin{equation}\label{eq:rank-bound}
q_r\le
\frac{2^{r-m}+1/(m^2N)+m^4/(8N^2)+N^{-2}}
     {1-m^2/(2N)},\qquad 1\le r\le m.
\end{equation}

\begin{proof}[Completion of the proof of Theorem~\ref{thm:main}]
Increasing $k$ increases the number $\binom Nk$ of candidate sets,
but can only increase their dependence probability $q_{k+1}$. The main term
$2^{r-m}$ in \eqref{eq:rank-bound} halves whenever one column is
removed. To bring it below $1/N$, we therefore need a shortfall
of about $\log_2N=n$ columns.

Put $s=n+2$ and $k=m-s$. We have $N2^{k+1-m}=1/2$ and
$m\ge\binom n2$, while \eqref{eq:auxiliary-bias-bounds} gives
$m^4/N\le2^{-0.994n}$. Hence \eqref{eq:rank-bound} gives
\[
 (N-k)q_{k+1}\le\frac12+o(1).
\]
By \eqref{eq:good-probability}, at least a proportion
$1/2-o(1)$ of the $(m-n-2)$-element sets are good. Hence
\eqref{eq:enumeration-transfer} gives
\begin{equation}\label{eq:half-lower}
 T_{n,d}\ge\left(\frac12-o(1)\right)\binom N{m-n-2}.
\end{equation}

Since $m^2/N\to0$, Lemma~\ref{lem:cover-estimates}(i) gives
\[
 Q_{n,d}=2(1+o(1))\binom N{m-1}.
\]

The ratio $\binom N{m-1}/\binom N{m-s}$ is a product of
$n+1=s-1$ factors. We will show that the logarithm of each
factor is approximately $\log_2(N/m)$, with errors that sum
to $O(1)$. More precisely,
\begin{align*}
 \log_2\frac{\binom N{m-1}}{\binom N{m-s}}
 &=\sum_{j=m-s+1}^{m-1}\log_2\frac{N-j+1}{j}\\
 &=(s-1)\log_2(N/m)
   +O\!\left(\frac{s^2}{m}+\frac{sm}{N}\right)\\
 &=(s-1)\log_2(N/m)+O(1).
\end{align*}
For the error estimate, the range $m-s+1\le j\le m-1$ gives
$1\le m-j\le n+1$ and $0\le j-1\le m-2$. We write
\[
\log_2\frac{N-j+1}{j}
=\log_2(N/m)
+\log_2\left(1-\frac{j-1}{N}\right)
-\log_2\left(1-\frac{m-j}{m}\right).
\]
Both $(j-1)/N$ and $(m-j)/m$ tend to zero uniformly over this
range.
The estimate $\log_2(1-u)=O(u)$ as $u\to0$ therefore gives
an error $O(m/N+(n+1)/m)$ per summand, uniformly over this
range. Summing the $n+1$ errors gives
\[
O\left(\frac{(n+1)m}{N}+\frac{(n+1)^2}{m}\right)=O(1).
\]
Here $m\ge\binom n2$ bounds the second term, while $m\ge n+1$
gives $(n+1)m/N\le m^2/N\to0$ for the first.
The factor preceding the binomial coefficient in
\eqref{eq:half-lower} has logarithm $-1+o(1)$. Including the
factor $2(1+o(1))$ in $Q_{n,d}$ proves
\[
 0\le\log_2Q_{n,d}-\log_2T_{n,d}
 \le(n+1)\log_2(N/m)+O(1)=O(n^2).
\]

Each column given up in this comparison costs approximately
$\log_2(N/m)$ in the logarithmic count. Lemma~\ref{lem:cover-estimates}(ii)
gives $\log_2Q_{n,d}\ge1+(m-1)\log_2(N/m)$, so dividing by
$\log_2Q_{n,d}$ cancels this factor and gives $E_{n,d}\le Cn/m$.
Here we use $\log_2(N/m)\ge1$ for all sufficiently large $n$.

Finally, Lemma~\ref{lem:cover-estimates}(ii) gives
\[
 \log_2Q_{n,d}=m\log_2(N/m)+O(m+n).
\]

By \eqref{eq:auxiliary-bias-bounds},
$\log_2(N/m)>0.9985n$, and $m\ge\binom n2$. Hence $m+n=O(m)$ and $n^2=O(m)$, while
$\log_2(N/m)\to\infty$. Both the $O(m+n)$ error in the estimate
for $\log_2Q_{n,d}$ and the $O(n^2)$ difference between
$\log_2Q_{n,d}$ and $\log_2T_{n,d}$ are therefore
$o(m\log_2(N/m))$. Consequently,
\[
\log_2T_{n,d}\sim\log_2Q_{n,d}\sim m\log_2(N/m),
\]
as required.
\end{proof}

The restriction to a small linear degree range enters in the weighted
bias estimates of Subsection~\ref{sec:large-bias}. Those estimates must
control almost the full sample dimension, rather than a fixed
fraction of it. With minor adjustments to the
auxiliary numerical bounds, the same estimates permit any fixed
$c<1/3808$, the restriction coming from $476\gamma<1/8$ in the
low-weight estimate (see Appendix~\ref{app:weight-estimates}).
We retain $10^{-4}$ for simplicity.
For the leading logarithm at larger degrees, we use the coding
estimates and deletion argument in Section~\ref{sec:all-degrees}.

\medskip
\noindent\emph{Remark.} The bound on $q_r$ also gives a statement about
random sets of evaluation vectors. For $s=n+\lceil\log_2n\rceil$
and $k=m-s$, \eqref{eq:rank-bound} gives
$(N-k)q_{k+1}\le2/n+o(1)$. Thus \eqref{eq:good-probability}
and \eqref{eq:enumeration-transfer} show that almost every set
of $m-s$ evaluation vectors is good and
$T_{n,d}\ge(1-o(1))\binom N{m-s}$.
This conclusion concerns samples with order $n$ fewer columns
than the dimension. We return to the question of reaching the full
dimension in Section~\ref{sec:conclusions}.

\section{The intermediate and near-central degrees}
\label{sec:all-degrees}

We now prove Theorems~\ref{thm:away-centre}
and~\ref{thm:near-central}. For the first, we use the weight bound of
Sberlo and Shpilka to show that a suitable random selection of
evaluation columns spans $\F_2^m$ with probability tending to one.
For the second, we bound the expected nullity of a random selection
of columns using the probability that a fixed column lies in the
span of the other selected columns. An entropy argument then gives
many independent sets. The deletion lemma converts the
independent-set counts in both arguments into lower bounds for
$T_{n,d}$.

\subsection{A quantitative bound below the central degree}
\label{sec:intermediate-degrees}

We derive from Theorem~1.2 of Sberlo and Shpilka
\cite{SberloShpilka2020} a bound on the probability that randomly
selected evaluation columns fail to span $\F_2^m$. This bound will
be used to prove Theorem~\ref{thm:away-centre}. The condition
$d/n\le1/2-\varepsilon$ keeps the degree proportion at least
$\varepsilon$ below $1/2$. Thus $\varepsilon\ge n^{-1/3}$ keeps
$d$ at least $n^{2/3}$ below the central degree.

\begin{lemma}[Full rank for a random evaluation matrix]
\label{lem:rank-retention}
There are absolute constants $C>0$ and $n_0$ such that the
following holds. Suppose $n\ge n_0$,
$n^{-1/3}\le\varepsilon<1/2$ and
$1\le d\le(1/2-\varepsilon)n$. Put $m=m_{n,d}$ and
\[
 D_\varepsilon=\frac{544\ln2}{\varepsilon},\qquad
 q=\frac{D_\varepsilon m}{2^n}.
\]
Then $0<q<1$. Choose $U\subseteq\{0,1\}^n$ by including each
vertex independently with probability $q$. The failure probability
satisfies
\begin{equation}\label{eq:retention-failure}
 \Pr\!\left(\rank_{\F_2}A_U<m\right)
 \le2^me^{-D_\varepsilon m/4}
 +d\,2^{Cn^4-D_\varepsilon m2^{-d}/(4\ln2)}.
\end{equation}
Consequently, if $d=d(n)$ and $\varepsilon=\varepsilon(n)$
satisfy the stated restrictions and $m2^{-d}/n^4\to\infty$,
then $A_U$ has full row rank with probability tending to one
as $n\to\infty$.

Moreover, for every fixed $c>0$ and every $\zeta>0$, there is
$n_1=n_1(c,\zeta)\ge n_0$ such that, whenever $n\ge n_1$ and
$cn\le d\le(1/2-\varepsilon)n$, the same construction satisfies
\[
 \Pr\!\left(\rank_{\F_2}A_U=m\right)\ge1-\zeta
\]
uniformly over $n^{-1/3}\le\varepsilon<1/2$.
\end{lemma}

The proof is given in Appendix~\ref{app:weight-estimates}.
Failure of full rank means that some non-zero codeword vanishes
at every selected vertex. A word of relative weight $w$ does so
with probability at most $e^{-D_\varepsilon mw}$, so words of small
weight present the greatest obstacle. The weight bound controls
their number in each weight range. The comparison involves factors
of order $\ell(2\beta)^{\ell-1}$, where
$\beta=1/2-\varepsilon$. These are bounded by
$1/(1-2\beta)=1/(2\varepsilon)$. This is why $D_\varepsilon$ is
chosen of order $1/\varepsilon$.

\begin{proof}[Proof of Theorem~\ref{thm:away-centre}]
Let $n^{-1/3}\le\varepsilon<1/2$ and put $m=m_{n,d}$.
All constants in this proof are absolute. If
$2\le d\le10^{-4}n$, Theorem~\ref{thm:main} bounds the additive
error by a constant multiple of $n^2$. Since $m\ge\binom n2$,
this is at most a constant multiple of $m$, and hence at most
a fixed constant times $m\ln(e/\varepsilon)$.
It remains to establish the additive bound for
$10^{-4}n<d\le(1/2-\varepsilon)n$ when this range is non-empty.
Recall that $J_r$ counts the independent $r$-element sets of
evaluation columns over $\F_2$. We first bound $J_m$ by counting
the bases contained in random samples, and then apply the deletion
bound \eqref{eq:central-deletion-count}.

Choose $U\subseteq V$ by including each vertex independently with
probability $q=D_\varepsilon m/N$. Taking $c=10^{-4}$ and
$\zeta=1/2$ in the final assertion of Lemma~\ref{lem:rank-retention},
we obtain, for all sufficiently large $n$,
\[
 \Pr\!\left(\rank_{\F_2}A_U=m\right)\ge\frac12.
\]
Whenever $A_U$ has rank $m$, it contains $m$ columns that are
linearly independent over $\F_2$. For any fixed $m$-element set
$S\subseteq V$, the probability that $S\subseteq U$ is $q^m$,
since all $m$ vertices must be selected independently. Thus the
union bound gives
\[
 \frac12\le\Pr\!\left(\rank_{\F_2}A_U=m\right)\le J_mq^m,
\]
so, using $q=D_\varepsilon m/N$,
\[
 \log_2J_m\ge m\log_2(N/m)-m\log_2D_\varepsilon-1.
\]
For our evaluation matrix, every dependent set has at least $2^{d+1}$
columns. The deletion bound \eqref{eq:central-deletion-count}, with
$r=m$, therefore gives
\[
 \log_2T_{n,d}\ge m\log_2(N/m)-m\log_2D_\varepsilon
 -O\!\left(\frac{mn^2}{2^d}+n\right).
\]
The final error term is $o(m)$ uniformly in this range. Indeed,
$d>10^{-4}n$ and $m\ge\binom n2$ give
\[
 \frac{mn^2/2^d+n}{m}
 =\frac{n^2}{2^d}+\frac nm
 \le n^2\,2^{-10^{-4}n}+\frac2{n-1}\longrightarrow0.
\]

Since $m\le N/2$, Lemma~\ref{lem:cover-estimates}(ii) gives
\[
 \log_2Q_{n,d}=m\log_2(N/m)+O(m+n).
\]
Since $D_\varepsilon=544\ln2/\varepsilon$ and $m\ge\binom n2$,
comparison with the lower bound gives
\[
 0\le\log_2Q_{n,d}-\log_2T_{n,d}
 \le K m\ln(e/\varepsilon),
\]
for a fixed absolute constant $K$ chosen large enough also to cover
the small-degree estimate. This proves the additive assertion
throughout the stated range.

Finally, the upper bound in \eqref{eq:binomial-entropy} gives
\[
 \log_2(N/m)\ge n\{1-H_2(1/2-\varepsilon)\}
 \ge\frac{2}{\ln2}n\varepsilon^2,
\]
where the second inequality follows from Taylor's theorem.
Since $m-1\ge m/2$, the lower bound in
Lemma~\ref{lem:cover-estimates}(ii) gives
$\log_2Q_{n,d}\ge mn\varepsilon^2/\ln2$. Dividing the
additive bound by this quantity gives
$E_{n,d}\le K\ln2\,\ln(e/\varepsilon)/(n\varepsilon^2)$.
Since $\ln2<1$, taking the theorem constant $C\ge K$ proves
both assertions.
\end{proof}

\subsection{Degrees just below the central degree}
\label{sec:central-counts}

We prove Theorem~\ref{thm:near-central}. Throughout this subsection,
suppose $1\le d<n/2$, and recall that
\[
 \alpha=\frac mN\le\frac12,\qquad
 a=\ln(1/\alpha).
\]
We assume that $a/n\le\delta_0$ and that $n$ is sufficiently
large. We choose the absolute constant $\delta_0>0$ small enough
to meet all the requirements below and in
Appendix~\ref{app:near-central-ii}, and use the same constant in
Theorem~\ref{thm:near-central}(i). The constants below are absolute,
unless stated otherwise.

For part~(i), we select each evaluation column independently with
probability $\alpha$, so that the sample contains $m$ columns on
average. We use a bound on the expected nullity, together with an
entropy argument, to count linearly independent sets over $\F_2$.
The deletion lemma then converts this into a count of good sets and
hence a lower bound on $T_{n,d}$.

Let $h(t)$ be the probability that a fixed column lies in the span
over $\F_2$ of the other columns selected independently with
probability $t$. By symmetry, this probability is the same for every
fixed column. An integral identity proved below expresses the
expected nullity at sampling probability $\alpha$ as $N$ times the
area under $h$ between $0$ and $\alpha$. If $h$ jumped from $0$ to
$1$ at a single point, a second identity would force that point to
be $\alpha$, and this area would be zero. A sharp-threshold
inequality shows that $h$ is close to such a step, so the area is
small. For part~(ii), we sample slightly above $\alpha$, to bring
the expected rank closer to $m$.

\subsubsection{An entropy count of independent sets}
\label{sec:entropy-count}

We will use the following count, which controls all sizes of
independent sets at once. For a random variable $X$ taking finitely
many values, its Shannon entropy is
\[
 H(X)=-\sum_x\Pr(X=x)\log_2\Pr(X=x),
\]
with $0\log_2 0=0$. We write $H(X,Y)$ for the entropy of the pair
$(X,Y)$. We use the standard facts that entropy is at most the
logarithm of the number of possible values, that entropy is
subadditive, and that $H_2$ is concave (see Cover and
Thomas~\cite[Chapter~2]{CoverThomas2006}).

\Needspace{12\baselineskip}
\begin{lemma}\label{lem:entropy-independent}
Let $B$ be a binary matrix with $N\ge1$ columns, and let $I(B)$ be the
number of subsets of its column indices whose corresponding columns
are independent over $\F_2$, including the empty set.
Choose a random set $U$ of column indices by
including each index independently with probability $q$. Then
\[
 \log_2 I(B)\ge NH_2(q)
 -NH_2\!\left(\frac{\E(|U|-\rank_{\F_2}B_U)}N\right).
\]
\end{lemma}

\begin{proof}
Fix an ordering of the columns of $B$. For each sampled set $U$,
construct $S(U)$ by examining its indices in that order, starting
with the empty set. Add an index precisely when its column does
not lie in the span over $\F_2$ of the columns already selected.

The selected columns are independent. Every skipped column lies in
the span of those selected at that stage, and hence in the span of
the final selection. Thus the selected columns form a basis for the
columns of $B_U$. Putting $R(U)=U\setminus S(U)$, we therefore have
\[
 |S(U)|=\rank_{\F_2}B_U,\qquad
 |R(U)|=|U|-\rank_{\F_2}B_U.
\]

Write $S=S(U)$ and $R=R(U)$. Since $(S,R)$ is a function of $U$
and $U=S\cup R$ is a function of $(S,R)$, the two random variables
have the same entropy:
\[
 H(U)=H(S,R).
\]

We may assume $0<q<1$, since the endpoint cases follow from
$I(B)\ge1$. The membership indicators defining $U$ are independent
Bernoulli variables, so $H(U)=NH_2(q)$. Thus, by subadditivity,
\[
 NH_2(q)=H(S,R)\le H(S)+H(R).
\]

Every possible value of $S$ is an independent set of column
indices, so $H(S)\le\log_2 I(B)$.

It remains to bound $H(R)$ in terms of its expected size.
Put $p_v=\Pr(v\in R)$. Since $|R|$ counts the indices belonging to
$R$, we have $\sum_{v=1}^N p_v=\E|R|$.
The membership indicator for index $v$ has entropy $H_2(p_v)$.
Subadditivity bounds the entropy of $R$ by the sum of these
individual entropies. Concavity of $H_2$ then bounds this sum by
$N$ times the binary entropy of the average membership probability.
Hence
\[
 H(R)\le\sum_{v=1}^N H_2(p_v)
 \le NH_2\!\left(\frac1N\sum_{v=1}^N p_v\right)
 =NH_2\!\left(\frac{\E|R|}N\right).
\]
Combining the bounds for $H(S)$ and $H(R)$, we obtain
\[
 NH_2(q)\le\log_2 I(B)+NH_2\!\left(\frac{\E|R|}N\right).
\]
Substituting $\E|R|=\E(|U|-\rank_{\F_2}B_U)$ and rearranging
proves the lemma.
\end{proof}

We now apply Lemma~\ref{lem:entropy-independent} to the evaluation
matrix $A$. Recall that $J_r$ counts its independent $r$-element
sets over $\F_2$, and put $I_{\mathrm{bin}}=\sum_{r=0}^mJ_r$,
where $J_0=1$. Thus $I_{\mathrm{bin}}$ counts its independent sets
of all sizes.

Choose $r$ for which $J_r$ is largest. There are $m+1$ possible
sizes, so $J_r\ge I_{\mathrm{bin}}/(m+1)$. Moreover, $r\ge1$
since $J_1=N>J_0$. For $d<n$, the deletion bound
\eqref{eq:central-deletion-count} gives
\[
 \log_2T_{n,d}\ge\log_2J_r-O\!\left(\frac{rn^2}{2^d}+n\right)
 \ge\log_2I_{\mathrm{bin}}-\log_2(m+1)
       -O\!\left(\frac{rn^2}{2^d}+n\right).
\]
Since $r\le m$ and $\log_2(m+1)\le\log_2(2^n+1)\le n+1$,
we obtain
\[
 \log_2T_{n,d}\ge\log_2I_{\mathrm{bin}}
       -O\!\left(\frac{mn^2}{2^d}+n\right).
\]
For any $0\le q\le1$, let $U\subseteq V$ include each vertex
independently with probability $q$.
Applying Lemma~\ref{lem:entropy-independent} with $B=A$ now yields
\begin{equation}\label{eq:entropy-transfer}
 \log_2T_{n,d}\ge NH_2(q)
 -NH_2\!\left(\frac{\E(|U|-\rank_{\F_2}A_U)}N\right)
 -O\!\left(\frac{mn^2}{2^d}+n\right).
\end{equation}
To use \eqref{eq:entropy-transfer}, we need to estimate the expected
nullity of $A_U$. An upper bound on the expected nullity divided by
$N$ can be substituted into $H_2$ provided that bound is at most
$1/2$, since $H_2$ is increasing on $[0,1/2]$. We will check this
condition when applying the estimate.

\subsubsection{Estimates for Theorem~\ref{thm:near-central}}
\label{sec:central-common-estimates}

Recall that $\Lambda=NH_2(\alpha)$. Lemma~\ref{lem:cover-estimates}(iii)
gives
\[
 \log_2Q_{n,d}=\Lambda+O(n),\qquad
 \frac\Lambda2\le\log_2Q_{n,d}\le\Lambda+1
\]
for all sufficiently large $n$, uniformly over $1\le d<n/2$.
Since $a\ge\ln2$, its entropy bounds also give
\[
 \frac{m(a+1)}{1+\ln2}\le\Lambda\le\frac{m(a+1)}{\ln2},
 \qquad \Lambda\ge m.
\]
Thus $\Lambda$ has order $m(a+1)$. It remains to show that the
deletion error in \eqref{eq:entropy-transfer} is negligible on this
scale.

We now use the assumption $a/n\le\delta_0$ to control the error
term in \eqref{eq:entropy-transfer}. If $d<n/4$, the upper
binomial bound in \eqref{eq:binomial-entropy} would give
\[
 a/n\ge\ln2\{1-H_2(1/4)\}>0.
\]
Choosing $\delta_0$ smaller than this constant therefore forces
$d\ge n/4$. Also,
$m=Ne^{-a}\ge e^{(\ln2-\delta_0)n}$. Since $\Lambda\ge m$,
\[
 \frac{mn^2/2^d+n}{\Lambda}
 \le\frac{n^2}{2^d}+\frac nm
 \le n^2 2^{-n/4}+n e^{-(\ln2-\delta_0)n}.
\]
Thus, for all sufficiently large $n$, the error term in
\eqref{eq:entropy-transfer}, divided by $\Lambda$, is at most
$e^{-c'n}$ for an absolute constant $c'>0$.
Consequently,
\[
 \frac{mn^2/2^d+n}{\log_2Q_{n,d}}=O(e^{-c'n}).
\]
We use this estimate in both parts of Theorem~\ref{thm:near-central}
without repeating the comparison. Since this error is $o(1/n)$
and $a\ge\ln2$, it can be absorbed into either stated relative-error
bound.

\subsubsection{Rank of a random set of columns}
\label{sec:coordinate-recovery}

We now define the function $h$ used in the overview and establish
the integral identities for expected rank and nullity.

Fix $v\in V$. Choose $U_t^{(v)}\subseteq V\setminus\{v\}$ by
including each vertex independently with probability $t$, and put
\[
 h_v(t)=\Pr\bigl(A_v\in\Span_{\F_2}
                    \{A_u:u\in U_t^{(v)}\}\bigr).
\]
Thus $h_v(t)$ is the probability that adding column $A_v$ to the
selected columns does not increase their rank.

An invertible affine map of $\F_2^n$ has the form
$\sigma(x)=Lx+b$, where $L$ is an invertible $n\times n$ matrix
over $\F_2$ and $b\in\F_2^n$. Substitution by $\sigma$ means
replacing a polynomial $p(x)$ by $p(\sigma(x))$. Each variable
is replaced by an affine linear form, so the degree cannot
increase. Thus this substitution preserves the space of binary
polynomial functions of degree at most $d$. Consequently, for each such map $\sigma$,
there is an invertible matrix $M_\sigma$ over $\F_2$ such that
\[
 A_{\sigma(x)}=M_\sigma A_x\qquad(x\in V).
\]
The matrix $M_\sigma$ is invertible because substitution by
$\sigma^{-1}$ gives its inverse. These maps therefore preserve
all linear relations among the columns.
Translations carry any vertex to any other and preserve the
distribution of the selected set, so $h_v(t)$ is independent of
$v$. We write $h(t)$ for its common value (see also
\cite[Proposition~7 and Lemma~23]{KudekarEtAl2017}).

We now express the expected rank and nullity of a random selection
of columns as integrals of $h$.

\Needspace{12\baselineskip}
\begin{lemma}[Integral identities for expected rank and nullity]
\label{lem:recovery-integrals}
For $0\le t\le1$, choose $U\subseteq V$ by including each vertex
independently with probability $t$, and put
$F(t)=\E\rank_{\F_2}A_U$. Then
\begin{equation}\label{eq:central-area-rank}
 \begin{aligned}
 \int_0^1(1-h(u))\,du&=\alpha,\\
 \E\bigl(|U|-\rank_{\F_2}A_U\bigr)&=N\int_0^t h(u)\,du,\\
 m-F(t)&=N\int_t^1(1-h(u))\,du.
 \end{aligned}
\end{equation}
\end{lemma}

\begin{proof}
For $S\subseteq V$, put
\[
 f(S)=|S|-\rank_{\F_2}A_S,
 \qquad G(t)=\E f(U).
\]
Adding column $A_v$, for $v\notin S$, increases the number of
columns by one. If it increases the rank, the nullity is unchanged,
but otherwise the nullity increases by one. Thus $f(S\cup\{v\})-f(S)$
is $1$ precisely when $A_v$ lies in the span of the columns indexed
by $S$, and is $0$ otherwise. We use this observation to calculate
$G'(t)$.

Since the vertices are selected independently,
\[
 G(t)=\sum_{W\subseteq V}f(W)t^{|W|}(1-t)^{N-|W|}.
\]
This is a polynomial in $t$, which we differentiate term by term.

The probability $t^{|W|}(1-t)^{N-|W|}$ is a product with one
factor for each vertex: $t$ for a vertex in $W$, and $1-t$ for
a vertex outside $W$. We apply the product rule and group the
resulting terms according to the vertex whose factor is
differentiated. Fix $v\in V$, and let
$S\subseteq V\setminus\{v\}$. The positive term associated with
$v$ comes from $W=S\cup\{v\}$, and the negative term comes
from $W=S$. Both have the factor
$t^{|S|}(1-t)^{N-1-|S|}$, with coefficients $f(S\cup\{v\})$
and $-f(S)$, respectively. Summing these contributions gives,
for $0<t<1$,
\[
 G'(t)=\sum_{v\in V}\ \sum_{S\subseteq V\setminus\{v\}}
 \bigl(f(S\cup\{v\})-f(S)\bigr)
 t^{|S|}(1-t)^{N-1-|S|}.
\]
The inner sum is the expectation of $f(S\cup\{v\})-f(S)$ when
$S$ is formed by selecting each vertex of $V\setminus\{v\}$
independently with probability $t$. This difference is $1$ when
$A_v$ lies in the span of the columns indexed by $S$, and $0$
otherwise. Its expectation is therefore the probability of that
event, namely $h_v(t)$. Hence
\[
 G'(t)=\sum_{v\in V}h_v(t)=Nh(t).
\]
As $G(0)=0$, integration proves the nullity identity. Moreover,
\[
 F(t)=Nt-G(t)=N\int_0^t(1-h(u))\,du.
\]
At $t=1$, all columns are selected and $F(1)=m$. Dividing the
last identity at $t=1$ by $N$ gives the area identity, and
subtracting its value at $t$ from its value at $1$ gives
\[
 m-F(t)=N\int_t^1(1-h(u))\,du.
\]
\end{proof}

The function $h$ also has a coding interpretation. In the dual
Reed--Muller code $\RM(n,n-d-1)=\ker_{\F_2}A$, it is the
probability of failure to recover a fixed erased coordinate when
each other coordinate is erased independently with probability $t$
(see Kudekar et al.~\cite[Definition~2 and Proposition~4(a)]{KudekarEtAl2017}).
Indeed, failure occurs precisely when there is a word $z$ in this
code with $z_v=1$ and support contained in $U_t^{(v)}\cup\{v\}$.
The equation $Az=0$ is then exactly the column dependence defining
$h_v(t)$. In this interpretation, the first identity in
\eqref{eq:central-area-rank} is equivalent to the area identity
in \cite[Proposition~4(c)]{KudekarEtAl2017}, and $G'(t)=Nh(t)$
also follows from Proposition~1 of that paper. Coordinate recovery
also underlies the approach in the author's centrality
paper~\cite{Anthony2026}.

For Theorem~\ref{thm:near-central}(i), the nullity identity reduces
the problem to bounding $\int_0^\alpha h(t)\,dt$. For part~(ii),
we sample at $t>\alpha$ and bound the rank deficiency through
$\int_t^1(1-h(u))\,du$. The following lemma supplies both bounds.

The affine symmetries fixing one vertex act transitively on the
other vertices, so each selection coordinate has the same influence
on the event defining $h$. O'Donnell's sharp-threshold inequality
for biased sampling measures then forces a rapid increase of $h$.
The factor $1/[t\ln(e/t)]$ in the differential inequality is
important here, since the sampling probability $\alpha$ may tend
to zero.

\begin{lemma}[A differential inequality and integral bounds]
\label{lem:recovery-bounds}
There is an absolute constant $c_2>0$ such that, with
$\kappa=c_2n$, the following hold. For $0<p<1$, put
$b_p=\ln(e/p)$.
\begin{enumerate}
\item[(i)] For $0<t<1$,
\begin{equation}\label{eq:uniform-threshold}
 \frac{d}{dt}\ln\frac{h(t)}{1-h(t)}
 \ge\frac{\kappa}{t\ln(e/t)}.
\end{equation}
\item[(ii)] If $\kappa>1$ and $h(p)\le1/2$, then
\begin{equation}\label{eq:recovery-left}
 \int_0^p h(t)\,dt\le\frac{p b_p}{\kappa-1}.
\end{equation}
\item[(iii)] Put $z_p=(1-h(p))/h(p)$. If $b_p<\kappa$, then
\begin{equation}\label{eq:recovery-tail}
 \int_p^1(1-h(t))\,dt\le\frac{z_ppb_p}{\kappa-b_p},
\end{equation}
and consequently
\begin{equation}\label{eq:recovery-rank-bound}
 m-F(p)\le\frac{Nz_ppb_p}{\kappa-b_p}.
\end{equation}
\end{enumerate}
\end{lemma}

The differential inequality follows from O'Donnell's
sharp-threshold inequality~\cite[Section~10.3, equation~(10.8)]{ODonnell2014},
using the affine symmetries of the evaluation columns over $\F_2$.
Integrating it gives the two integral bounds. The detailed proof
is given in Appendix~\ref{app:recovery-estimates}.

\begin{lemma}[Expected nullity at sampling probability $\alpha$]
\label{lem:central-nullity}
Recall that the evaluation matrix $A$ has $m$ rows and
$N=2^n$ columns. Put $\alpha=m/N$ and
$a=\ln(1/\alpha)=\ln(N/m)$.
There are absolute constants $\delta_0,C_I>0$ and $n_0$ such that,
whenever $n\ge n_0$, $1\le d<n/2$ and $a\le\delta_0n$,
\begin{equation}\label{eq:central-nullity}
 \int_0^\alpha h(t)\,dt
   =\int_\alpha^1(1-h(t))\,dt
   \le C_I\alpha\frac an.
\end{equation}
Consequently, if $U\subseteq V$ includes each vertex independently
with probability $\alpha$, then
\[
 \E\bigl(|U|-\rank_{\F_2}A_U\bigr)
 =m-\E\rank_{\F_2}A_U
 \le C_Im\frac an.
\]
\end{lemma}

\begin{proof}
The equality of the two integrals follows from the first identity
in \eqref{eq:central-area-rank}. Write $I$ for their common value.
To bound $I$, we may therefore use either integral.
We apply Lemma~\ref{lem:recovery-bounds} with $p=\alpha$, so
\[
 b_\alpha=\ln(e/\alpha)=a+1.
\]

If $h(\alpha)\le1/2$, part~(ii) of
Lemma~\ref{lem:recovery-bounds} bounds the
first integral:
\[
 I\le\frac{\alpha(a+1)}{\kappa-1}.
\]
If $h(\alpha)\ge1/2$, then
$(1-h(\alpha))/h(\alpha)\le1$, so part~(iii) bounds the
second integral:
\[
 I\le\frac{\alpha(a+1)}{\kappa-a-1}.
\]

Recall that $\kappa=c_2n$. Since $a\le\delta_0n$, we can
choose $\delta_0$ sufficiently small and $n_0$ sufficiently
large that both denominators are at least $c_2n/2$.
These choices also ensure the hypotheses of Lemma~\ref{lem:recovery-bounds}(ii) and (iii).
Thus, in either case,
\[
 I\le\frac{2\alpha(a+1)}{c_2n}
   \le C_I\alpha\frac an,
\]
where the last inequality uses $a\ge\ln2$.

By the second identity in \eqref{eq:central-area-rank}, the
expected nullity is $NI\le C_Im a/n$. Since $\E|U|=N\alpha=m$,
the expected nullity also equals $m-\E\rank_{\F_2}A_U$.
\end{proof}

\subsubsection{Proof of Theorem~\ref{thm:near-central}(i)}

\begin{proof}
Choose $U\subseteq V$ by including each vertex independently
with probability $\alpha$, so that $\E|U|=m$. Taking
$\delta_0\le\min\{1,1/C_I\}$, Lemma~\ref{lem:central-nullity} gives
\[
 \frac{\E(|U|-\rank_{\F_2}A_U)}N
 \le C_I\alpha\frac an\le\frac12.
\]
We may therefore use this upper bound in the increasing part of
$H_2$. The bound \eqref{eq:entropy-transfer} yields
\[
 \log_2T_{n,d}\ge
 \Lambda-NH_2\!\left(C_I\alpha\frac an\right)
 -O\!\left(\frac{mn^2}{2^d}+n\right).
\]

Using $H_2(x)\le x\ln(e/x)/\ln2$, $N\alpha=m$, and
$\ln(1/\alpha)=a$, we obtain
\[
\begin{aligned}
 NH_2\!\left(C_I\alpha\frac an\right)
 &\le \frac{C_I}{\ln2}\,m\frac an
       \{1+a+\ln(n/a)-\ln C_I\}\\
 &\le K m\frac an\{a+\ln(n/a)+1\}
\end{aligned}
\]
for an absolute constant $K>0$.
By Lemma~\ref{lem:cover-estimates}(iii) and
Subsection~\ref{sec:central-common-estimates},
$\log_2Q_{n,d}\le\Lambda+1$ and
$\log_2Q_{n,d}\ge c\,m(a+1)$ for an absolute $c>0$, while the
normalised deletion error is $O(e^{-c'n})$. Consequently,
\[
\begin{aligned}
 E_{n,d}
 &\le \frac{NH_2(C_I\alpha a/n)}{\log_2Q_{n,d}}+O(e^{-c'n})\\
 &\le \frac Kc\frac an
        \left\{1+\frac{\ln(n/a)}{a+1}\right\}+O(e^{-c'n})\\
 &\le \frac Kc\frac{a+\ln(n/a)}n+O(e^{-c'n}).
\end{aligned}
\]
The final inequality uses $a/n\le\delta_0\le1$ and
$a/(a+1)\le1$. Absorbing the exponentially small term as noted in
Subsection~\ref{sec:central-common-estimates}, and using
$T_{n,d}\le Q_{n,d}$, proves
\[
 0\le E_{n,d}\le C\frac{a+\ln(n/a)}n
\]
for an absolute $C>0$.

Finally, if $n/2-d=o(n)$, then $a=o(n)$, and
\[
 \frac an\longrightarrow0,\qquad
 \frac{\ln(n/a)}n\le\frac{\ln(n/\ln2)}n\longrightarrow0.
\]
Thus $E_{n,d}\to0$. Together with
$\log_2Q_{n,d}\sim\Lambda$ from
Lemma~\ref{lem:cover-estimates}(iii), this gives
\[
 \log_2T_{n,d}\sim\log_2Q_{n,d}\sim NH_2(\alpha).
\]
\end{proof}

\subsubsection{Sampling more columns for Theorem~\ref{thm:near-central}(ii)}

The proof of Theorem~\ref{thm:near-central}(ii) is given in
Appendix~\ref{app:near-central-ii}. Here we introduce the sampling
choice and state the bound on the expected rank of the selected
matrix.

Recall that $\alpha=m/N$ and $a=\ln(1/\alpha)$. In part~(i),
we selected $m$ columns on average. For part~(ii), we select
slightly more columns, with the aim of making their expected rank
closer to $m$.

Choose $U\subseteq V$ by including each vertex independently with
probability
\[
 q=(1+\eta)\alpha,\qquad 0<\eta\le1/10.
\]
Thus $\E|U|=(1+\eta)m$: the parameter $\eta$ specifies the
additional $\eta m$ columns selected on average. Since $A_U$ has $m$ rows, its rank is at most $m$. Define
\[
 \rho=\frac{m-\E\rank_{\F_2}A_U}{m}\ge0.
\]
Thus $\rho$ is the difference between $m$ and the expected rank,
as a proportion of $m$, and
$\E\rank_{\F_2}A_U=(1-\rho)m$. Consequently,
\[
 \E(|U|-\rank_{\F_2}A_U)
 =(1+\eta)m-(1-\rho)m=m(\eta+\rho).
\]
The next lemma bounds $\rho$. We then use this expression for the
expected nullity in the counting bound \eqref{eq:entropy-transfer}.

\begin{lemma}[Expected rank]
\label{lem:retention-excess}
There are absolute constants $\delta_0,\mu,C_0>0$ and $n_0$ such
that the following holds. Suppose $n\ge n_0$, $1\le d<n/2$,
$a\le\delta_0n$ and $C_0a/n\le\eta\le1/10$. Choose
$U\subseteq V$ by including each vertex independently with
probability $q=(1+\eta)\alpha$, and put
\[
 \rho=\frac{m-\E\rank_{\F_2}A_U}{m}.
\]
Then
\[
 0\le\rho\le C_0\frac an e^{-\mu n\eta/a}\le\eta.
\]
\end{lemma}

The proof of Lemma~\ref{lem:retention-excess} is given in
Appendix~\ref{app:near-central-ii}. To apply the lemma, recall that
$\E(|U|-\rank_{\F_2}A_U)=m(\eta+\rho)$ and $m/N=\alpha$.
Substituting $q=(1+\eta)\alpha$ in \eqref{eq:entropy-transfer}
therefore gives
\[
 \log_2T_{n,d}\ge
 NH_2((1+\eta)\alpha)-NH_2(\alpha(\eta+\rho))
 -O\!\left(\frac{mn^2}{2^d}+n\right).
\]
We compare this with $\log_2Q_{n,d}\le NH_2(\alpha)+1$,
combining the sampling entropy and nullity penalty so that the
contributions $\eta ma/\ln2$ from the extra columns cancel.
Appendix~\ref{app:near-central-ii} then uses the exponential bound
on $\rho$ in Lemma~\ref{lem:retention-excess} and chooses $\eta$
to prove Theorem~\ref{thm:near-central}(ii).

\subsubsection{Consequences of Theorem~\ref{thm:near-central}}

\Needspace{4\baselineskip}
The two parts of Theorem~\ref{thm:near-central} complement one
another. When $a=\Theta((\ln n)^2)$, part~(ii) gives
$E_{n,d}=O(\ln n\,\ln\ln n/n)$, whereas part~(i) gives
$O((\ln n)^2/n)$. When $a=\Theta(\ln n)$, part~(i) is sharper.

For odd $n$, Theorem~\ref{thm:near-central} includes the lower
central degree $\lfloor n/2\rfloor$. The upper central degree
$\lceil n/2\rceil$ is covered by the author's stronger centrality
theorem~\cite[Theorem~1]{Anthony2026}.

Fix $K>0$ and consider degrees satisfying
\[
 0<n/2-d\le K\sqrt n.
\]
In this window, the binomial probability $\alpha$ is bounded
below by a positive constant depending only on $K$.
Theorem~\ref{thm:near-central}(i) therefore gives
\begin{equation}\label{eq:central-window-rate}
 E_{n,d}\le C_K\frac{\ln n}{n}
\end{equation}
for all sufficiently large $n$, where $C_K$ depends only on $K$
and the bound holds for every degree in the window.

To illustrate the improvement supplied by
Theorem~\ref{thm:near-central}(ii), take
$d=\lfloor n/2-n^{3/4}\rfloor$. The bounds in
\eqref{eq:binomial-entropy} and expansion about $1/2$ give
$a\sim2\sqrt n$. For this degree, part~(i) gives
$E_{n,d}=O(n^{-1/2})$, whereas part~(ii) gives
\begin{equation}\label{eq:moderate-window-rate}
 E_{n,d}=O((\ln n)^2/n).
\end{equation}

We can make the asymptotic formula for $\log_2T_{n,d}$ more
explicit when $d$ is within $O(\sqrt n)$ of $n/2$.
Suppose that $d<n/2$ and
$d=n/2-\tau\sqrt n+o(\sqrt n)$ for a fixed $\tau\ge0$. If $B_n$ is
binomial with parameters $n$ and $1/2$, then
\[
 \alpha=\Pr(B_n\le d)\longrightarrow\Phi(-2\tau),
\]
where $\Phi$ is the standard normal distribution function
\cite{Feller1968}. Theorem~\ref{thm:near-central} therefore gives
\[
 \log_2T_{n,d}\sim2^nH_2(\Phi(-2\tau)).
\]
For example, when $\tau=1$, the coefficient $H_2(\Phi(-2))$ is
approximately $0.157$.

When $n/2-d=O(\sqrt n)$, so that $\alpha$ is bounded away from $0$,
the block-recovery theorem of
Kudekar et al.~\cite[Theorem~27]{KudekarEtAl2017}, applied to the
dual code $\ker_{\F_2}A$ and combined with
\eqref{eq:entropy-transfer}, also gives
$\log_2T_{n,d}\sim NH_2(\alpha)$, without a rate.
Theorem~\ref{thm:near-central} supplies the explicit bounds and
also covers $\alpha\to0$ with $a=o(n)$.

Related recursive recovery estimates for nested Reed--Muller
codes on the erasure channel appear in Pfister and
Reeves~\cite[Section~III]{PfisterReeves2025}.

\section{Deduction of the all-degree theorem}
\label{sec:all-degree-deduction}

We first verify the degree ranges stated after
Theorem~\ref{thm:near-central}, and then combine
Theorems~\ref{thm:main}, \ref{thm:away-centre}, and
\ref{thm:near-central} with the degree-one estimate to prove
Theorem~\ref{thm:all-degrees}.

\subsection{The degree ranges in Theorem~\ref{thm:near-central}}
\label{sec:central-degree-ranges}

For $1\le d<n/2$, recall that $a=\ln(2^n/m_{n,d})$.
Applying \eqref{eq:binomial-entropy} to $m=m_{n,d}$ gives
\[
 n\ln2\{1-H_2(d/n)\}\le a
 \le n\ln2\{1-H_2(d/n)\}+\ln(n+1).
\]
If $n/2-d=o(n)$, then $d/n\to1/2$ and $H_2(d/n)\to1$, so the
upper bound gives $a=o(n)$. Conversely, if $a=o(n)$, the lower
bound forces $H_2(d/n)\to1$. On $[0,1/2]$, the entropy approaches
$1$ only when its argument approaches $1/2$. Thus $d/n\to1/2$,
which means $n/2-d=o(n)$.

We now use these bounds to describe the degree ranges in
Theorem~\ref{thm:near-central}. We first record the entropy estimates
needed below. Since
\[
 H_2(1/2)=1,\qquad H_2'(1/2)=0,\qquad
 H_2''(u)=-\frac{1}{(\ln2)u(1-u)}\le-\frac4{\ln2},
\]
Taylor's theorem gives
\[
 1-H_2(u)\ge\frac{2}{\ln2}(u-1/2)^2
 \qquad(0<u<1).
\]
Also, $H_2$ is symmetric about $1/2$, so its odd derivatives there
vanish. Taylor's theorem therefore gives
\begin{equation}\label{eq:entropy-centre-expansion}
 H_2(1/2-\varepsilon)
 =1-\frac{2\varepsilon^2}{\ln2}+O(\varepsilon^4)
 \qquad(\varepsilon\to0).
\end{equation}

For part~(i), we require $a\le\delta_0n$. Since $H_2$ is continuous
and $H_2(1/2)=1$, we can choose a fixed $c_0>0$ sufficiently small
that
\[
 (\ln2)\{1-H_2(1/2-c_0)\}<\frac{\delta_0}{2}.
\]
Whenever $n/2-c_0n\le d<n/2$, monotonicity of $H_2$ on $[0,1/2]$
and the upper bound for $a$ give
\[
 \frac an
 \le(\ln2)\{1-H_2(d/n)\}+\frac{\ln(n+1)}n
 \le\frac{\delta_0}{2}+\frac{\ln(n+1)}n
 \le\delta_0
\]
for all sufficiently large $n$. Thus part~(i) applies throughout
this degree range.

To see where its hypothesis must fail, the lower bound for $a$ gives
\[
 a\ge n\ln2\{1-H_2(d/n)\}
 \ge\frac{(n-2d)^2}{2n}.
\]
If $d\le n/2-cn$, then $a\ge2c^2n$. Hence the hypothesis fails
whenever $c>\sqrt{\delta_0/2}$. These bounds show that the range of
part~(i) extends a distance proportional to $n$ below the centre.

For part~(ii), we require
\[
 \frac an\ln(a+2)\le\delta_1.
\]
Fix $0<c_1<\sqrt{\delta_1/2}$. We check this condition throughout
\[
 \frac n2-\frac{c_1n}{\sqrt{\ln n}}\le d<\frac n2.
\]
Since $a$ decreases as $d$ increases, and $a\ln(a+2)$ increases with
$a>0$, it is enough to check the smallest integer degree in this
range. At that degree,
\[
 \frac dn=\frac12-\frac{c_1}{\sqrt{\ln n}}+O(1/n).
\]
Combining the binomial bounds with the expansion
\eqref{eq:entropy-centre-expansion} gives
\[
 a=\frac{2c_1^2n}{\ln n}
   +O\!\left(\frac{n}{(\ln n)^2}+\ln n\right)
 \sim\frac{2c_1^2n}{\ln n}.
\]
It follows that $\ln(a+2)\sim\ln n$, and hence
\[
 \frac an\ln(a+2)\longrightarrow2c_1^2<\delta_1.
\]
The hypothesis therefore holds throughout the stated range for all
sufficiently large $n$.

The same calculation shows that, for any fixed
$c>\sqrt{\delta_1/2}$, the hypothesis fails throughout
$d\le n/2-cn/\sqrt{\ln n}$ for all sufficiently large $n$.
Thus the range of part~(ii) extends a distance of order
$n/\sqrt{\ln n}$ below the centre.

\subsection{Proof of Theorem~\ref{thm:all-degrees}}
\label{sec:all-degree-proof}

\begin{proof}[Proof of Theorem~\ref{thm:all-degrees}]
For $d=1$, the quantitative degree-one estimate of Kahn, Koml\'os
and Szemer\'edi~\cite[Section~4]{KahnKomlosSzemeredi1995},
recalled above, gives
\[
 \log_2Q_{n,1}-\log_2T_{n,1}=O(n).
\]
Since $\log_2Q_{n,1}\sim n^2$, it follows that
$E_{n,1}=O(1/n)$. We may therefore suppose that $d\ge2$.

With $c_1$ as in Subsection~\ref{sec:central-degree-ranges}, put
\[
 \varepsilon_n=\frac{c_1}{\sqrt{\ln n}}.
\]
All estimates below hold for sufficiently large $n$, with absolute
constants independent of $d$.

First suppose $2\le d\le(1/2-\varepsilon_n)n$.
Theorem~\ref{thm:away-centre} requires
$n^{-1/3}\le\varepsilon_n<1/2$. These inequalities hold for all
sufficiently large $n$, since $\varepsilon_n\to0$ and
\[
 n^{1/3}\varepsilon_n
 =\frac{c_1n^{1/3}}{\sqrt{\ln n}}\longrightarrow\infty.
\]
Its degree condition holds by our choice of range. Hence
\[
 E_{n,d}\le\frac{C\ln(e/\varepsilon_n)}{n\varepsilon_n^2}
 =O\!\left(\frac{\ln n\,\ln\ln n}{n}\right),
\]
where we used $\varepsilon_n^2=c_1^2/\ln n$ and
$\ln(e/\varepsilon_n)=1+\tfrac12\ln\ln n-\ln c_1$.

Now suppose $(1/2-\varepsilon_n)n<d<n/2$. Since
$\varepsilon_n=c_1/\sqrt{\ln n}$, this lies in the degree range
for Theorem~\ref{thm:near-central}(ii) established in
Subsection~\ref{sec:central-degree-ranges}. Recall that
$a=\ln(2^n/m_{n,d})$. Hence that theorem gives
\[
 E_{n,d}\le\frac{C\{1+\ln(a+2)\ln(n/a)\}}n.
\]
Since $m\ge1$, we have $a\le n\ln2$, so $a+2\le n$ and hence
$\ln(a+2)\le\ln n$ for all sufficiently large $n$. Also, $d<n/2$ implies
$m_{n,d}\le2^{n-1}$, so $a\ge\ln2$ and
\[
 \ln(n/a)\le\ln n-\ln(\ln2)=\ln n+O(1).
\]
Consequently,
\[
 E_{n,d}\le\frac{C\{1+\ln n\,(\ln n+O(1))\}}n
 =O\!\left(\frac{(\ln n)^2}{n}\right).
\]
Since $\ln n\,\ln\ln n\le(\ln n)^2$ for $n\ge3$, the bound
from Theorem~\ref{thm:away-centre} is also
$O((\ln n)^2/n)$. This therefore bounds $E_{n,d}$ uniformly
throughout $2\le d<n/2$.

Finally, suppose $d\ge\lceil n/2\rceil$, and put
$d_0=\lfloor n/2\rfloor-1$. By symmetry of the binomial coefficients and the central binomial
estimate, $\alpha_0:=m_{n,d_0}/N=1/2-O(n^{-1/2})$.
Since $H_2(1/2)=1$ and $H_2'(1/2)=0$, Taylor's theorem gives
\[
 1-H_2(\alpha_0)
 =O\!\left((\alpha_0-1/2)^2\right)=O(1/n).
\]
Since $\ln(1/\alpha_0)=O(1)$, Theorem~\ref{thm:near-central}(i)
applies for all sufficiently large $n$ and gives
$E_{n,d_0}=O(\ln n/n)$.
Lemma~\ref{lem:cover-estimates}(iii) gives
\[
 \log_2Q_{n,d_0}=NH_2(\alpha_0)+O(n).
\]
The last two estimates give
$\log_2T_{n,d_0}\ge N\{1-O(\ln n/n)\}$.
Monotonicity in the degree now yields
\[
 N\{1-O(\ln n/n)\}\le\log_2T_{n,d_0}
 \le\log_2T_{n,d}\le\log_2Q_{n,d}\le N.
\]
Thus $E_{n,d}=O(\ln n/n)$ uniformly in the upper range.

Combining the degree-one estimate with the bounds in the three
ranges, and increasing the absolute constants $C$ and $n_0$ if
necessary, gives
\[
 0\le E_{n,d}\le C\frac{(\ln n)^2}{n}
 \qquad(n\ge n_0,\ 1\le d\le n).
\]
Since $(\ln n)^2/n\to0$, the asserted uniform logarithmic
equivalence follows.
\end{proof}

The author's centrality theorem~\cite[Theorem~1]{Anthony2026}
gives a substantially stronger bound in the upper range, although
it is not needed for the proof of Theorem~\ref{thm:all-degrees}.
More precisely, for $d\ge\lceil n/2\rceil$, it gives
\[
 T_{n,d}\ge
 \left\{1-\exp\!\left(-c\frac{2^{n/2}}n\right)\right\}2^N
\]
for an absolute $c>0$. Since $T_{n,d}\le Q_{n,d}\le2^N$,
the definition of $E_{n,d}$ yields
\[
 0\le E_{n,d}
 \le-\frac1N\log_2\left\{1-\exp\!\left(-c\frac{2^{n/2}}n\right)\right\}
 \le C2^{-n}\exp\!\left(-c\frac{2^{n/2}}n\right).
\]

\section{Conclusions and open questions}
\label{sec:conclusions}

We have established the logarithmic form of the enumeration
conjecture: Theorem~\ref{thm:all-degrees} shows that
$\log_2T_{n,d}\sim\log_2Q_{n,d}$ for every degree sequence,
with relative error $O((\log n)^2/n)$ uniformly over all degrees.
The author's earlier centrality theorem~\cite[Theorem~1]{Anthony2026} already
establishes the stronger equivalence $T_{n,d}\sim Q_{n,d}$ for
$d\ge\lceil n/2\rceil$. The results here complete the logarithmic
asymptotic below the central degree.

Theorem~\ref{thm:main} gives the sharper additive gap $O(n^2)$
throughout $2\le d\le10^{-4}n$.
Theorem~\ref{thm:away-centre} gives relative error $O(1/n)$ for
$2\le d\le(1/2-\varepsilon)n$, for each fixed $0<\varepsilon<1/2$.
Theorem~\ref{thm:near-central} gives relative error $O(\log n/n)$
in the window $0<n/2-d=O(\sqrt n)$, together with estimates
throughout $n/2-d=o(n)$.

\Needspace{9\baselineskip}
The stronger question concerns the count itself. Does
\[
 T_{n,d}\sim Q_{n,d}
\]
hold for every sequence $2\le d<n/2$? Together with the known results
at degree one and at or above $\lceil n/2\rceil$, this would establish
the full enumeration conjecture. The logarithmic equivalence proved
here allows a multiplicative discrepancy that is still large, and so
does not resolve this question.

It would also be interesting to determine how far the additive
$O(n^2)$ estimate extends beyond the small linear degree range, and
whether its order can be improved within that range. The
$O(m_{n,d})$ bound of Theorem~\ref{thm:away-centre}, for each fixed
$\varepsilon$, becomes much larger
than $n^2$ when the degree grows, so it leaves this question open.
The constant
$10^{-4}$ is a convenient choice for the present weight estimates.
The proof does not suggest a transition at this degree. More precise
control of the number of evaluation vectors in a random span could
give sharper counting estimates.

The order of the relative error also remains open. Can the
logarithmic factors in the uniform bound of
Theorem~\ref{thm:all-degrees} be reduced? The logarithmic product
in \eqref{eq:degree-dependent-error} is maximised near $a=\sqrt n$,
corresponding to $n/2-d$ of order $n^{3/4}$. This identifies a
window in which to seek a sharper bound, but does not show that
the actual error is largest there. The present estimates do not
determine the best rate throughout $n/2-d=o(n)$.

Finally, consider the evaluation-matrix formulation of Luh's
question~\cite[Section~1.4 and Lemma~5.1]{Luh2026} about the dimension
threshold. His tensor problem in $n+1$ coordinates corresponds to
evaluation vectors in $n$ variables, with $D(n+1,d)=m_{n,d}$ in his
notation. For which growing degrees in
the range $d\le10^{-4}n$ are
$m_{n,d}$ independently and uniformly sampled evaluation vectors
linearly independent over $\R$ with probability tending to one?
Our argument uses order $n$ fewer vectors than the dimension.
Independence of a full sample would be a stronger rank statement,
although an additional argument would still be needed to obtain a
sharper enumeration bound.

\appendix
\section{Weight estimates and full rank}
\label{app:weight-estimates}

\subsection{Verification of the weight and bias estimates}

We retain the notation and hypotheses of
Subsection~\ref{sec:large-bias}: $c=10^{-4}$, $2\le d\le cn$,
$\gamma=d/n$, $W=W_{n,d}$ and $L=\lceil3\log_2m\rceil$.

\begin{proof}[Proof of Lemma~\ref{lem:weight-estimates}]
\leavevmode\par
\smallskip\noindent\emph{Verification of \eqref{eq:low-weight-input}.}
Corollary~3.6 and Proposition~3.7 of
Sberlo--Shpilka~\cite{SberloShpilka2020} give, for $2\le j\le d$,
\[
 \log_2W(2^{-j})
 \le17n\sum_{i=j}^d(i-1)(i+2)
       +17\sum_{i=j}^d(i+2)m_{n-i+1,d-i+1}.
\]
For $1\le k\le h\le d$,
\[
\frac{\binom{n-k}{h-k}}{\binom nh}
=\prod_{a=0}^{k-1}\frac{h-a}{n-a}
\le\left(\frac hn\right)^k
\le\gamma^k,
\]
since $(h-a)/(n-a)\le h/n$ when $h\le n$ and $0\le a<k$.
By the definition of $m_{n-k,d-k}$, changing the summation index
and applying this inequality term by term gives
\[
m_{n-k,d-k}
=\sum_{h=k}^{d}\binom{n-k}{h-k}
\le\gamma^k\sum_{h=k}^{d}\binom nh
\le\gamma^k m.
\]
The last inequality extends the sum to $h=0,\ldots,d$, adding
non-negative terms. Taking $k=i-1$ gives
$m_{n-i+1,d-i+1}\le\gamma^{i-1}m$.
Putting $t=i-j+1$, we have $i+2=j+t+1\le(j+2)t$, and hence
\[
\sum_{i=j}^{d}(i+2)\gamma^{i-1}
\le(j+2)\gamma^{j-1}\sum_{t=1}^{\infty}t\gamma^{t-1}
=\frac{(j+2)\gamma^{j-1}}{(1-\gamma)^2}
\le(2j+3)\gamma^{j-1}.
\]
The equality follows by differentiating the geometric series.
For the last inequality, use $(1-\gamma)^{-2}\le3/2$, since
$\gamma\le10^{-4}$, and $\tfrac32(j+2)\le2j+3$.
For the first sum, $(i-1)(i+2)\le2i^2$ gives
\[
17n\sum_{i=j}^d(i-1)(i+2)
\le34n\sum_{i=1}^d i^2\le34nd^3.
\]
Combining the bounds on the two sums gives
\[
\log_2W(2^{-j})\le34nd^3+17(2j+3)\gamma^{j-1}m.
\]
To bound the second term, put $a_j=17(2j+3)2^j\gamma^{j-1}$
for $j\ge2$. Then
\[
\frac{a_{j+1}}{a_j}
=2\gamma\frac{2j+5}{2j+3}
\le\frac{18}{7}\gamma<1,
\]
since $\gamma\le10^{-4}$. Thus $a_j\le a_2=476\gamma<1/8$,
and hence $17(2j+3)\gamma^{j-1}\le2^{-j}/8$.
It remains to show that $34nd^3=o(m2^{-d})$, uniformly over the
stated degree range.
The quantity $2^{-d}\binom nd/d^3$ increases for $2\le d\le cn$
and large $n$, since its ratio at consecutive degrees is
\[
 \frac{n-d}{2(d+1)}\left(\frac d{d+1}\right)^3
 \ge\frac{1-c}{4c}\left(\frac23\right)^3>1.
\]
Here $d+1\le2cn$ and $d/(d+1)\ge2/3$.
Comparing with $d=2$ and using $m\ge\binom nd$ gives
\[
\frac{m2^{-d}}{nd^3}
\ge\frac{2^{-d}\binom nd}{nd^3}
\ge\frac{n-1}{64}\longrightarrow\infty.
\]
Since $34nd^3=o(m2^{-d})$, this gives
\eqref{eq:low-weight-input} for $2\le j\le d$.
For $j>d$, this holds because the minimum-weight bound
\eqref{eq:minimum-weight} gives $W(2^{-j})=1$.

\smallskip\noindent\emph{Verification of \eqref{eq:auxiliary-bias-bounds}.}
For the inequality $e^{-24\gamma L}\ge2^{-d}$, it suffices to show that
$24\gamma L\le d\ln2$. The bound $m\le(en/d)^d$ and the fact
that $t\log_2(e/t)$ is increasing for $0<t<1$ give
\[
\log_2m\le n\gamma\log_2(e/\gamma)
\le nc\log_2(e/c)<0.0015n.
\]
Since $L=\lceil3\log_2m\rceil$, we have
\[
L\le3\log_2m+1<0.0045n+1<0.005n
\]
for all sufficiently large $n$. As $\gamma=d/n$,
\[
24\gamma L<24\frac dn(0.005n)=0.12d<d\ln2.
\]
The uniform comparison $nd^3=o(m2^{-d})$ established in the
verification of \eqref{eq:low-weight-input} implies $n=o(m2^{-d})$,
proving the uniform limit in \eqref{eq:auxiliary-bias-bounds}.

\smallskip\noindent\emph{Verification of \eqref{eq:bias-tail}.}
We use the following consequence of Corollary~3.13 and
Proposition~3.14 of Sberlo and Shpilka~\cite{SberloShpilka2020}.
For positive integers $\ell,s$, put $t=2\ell+s+1$.
If $t<n$ and $d<n-t$, then, in our notation,
\[
W\!\left(\frac{1-2^{-\ell}}2\right)
\le 2^{nt+\left[1-\left(1-\frac d{n-t}\right)^t\right]m}
W(2^{1-s}).
\]
For $1\le\ell\le L$, we take $s=\ell+2$, so that $t=3\ell+3$.
The bound $L<0.005n$, proved in the verification of
\eqref{eq:auxiliary-bias-bounds}, gives $t\le3L+3<n/2$ for all
sufficiently large $n$. Since $d\le10^{-4}n<n/2<n-t$, the
cited bound applies.

To estimate its exponent, we use $d/(n-t)\le2\gamma<1/2$,
$t\le6\ell$, and $\ln(1-u)\ge-2u$ for $0\le u\le1/2$.
These give
\[
\left(1-\frac d{n-t}\right)^t\ge e^{-24\gamma\ell}.
\]
Substituting this estimate and $2^{1-s}=2^{-\ell-1}$ into
the bound from Sberlo and Shpilka gives
\begin{equation}\label{eq:weight-reduction-input}
 W\!\left(\frac{1-2^{-\ell}}2\right)
 \le2^{nt+(1-e^{-24\gamma\ell})m}W(2^{-\ell-1}).
\end{equation}
By \eqref{eq:low-weight-input} and $24\gamma<\ln2$,
\[
 \log_2W(2^{-\ell-1})
 \le\frac18m2^{-\ell}\le\frac18m e^{-24\gamma\ell}.
\]
Since $\ell\le L$, \eqref{eq:auxiliary-bias-bounds} gives
$e^{-24\gamma\ell}\ge2^{-d}$ throughout this range.
Here $t\le3L+3=O(d\log_2n)$. The comparison established in
the verification of \eqref{eq:low-weight-input} gives
\[
\frac{nd\log_2n}{m2^{-d}}
\le\frac{64\log_2n}{(n-1)d^2}
\le\frac{16\log_2n}{n-1}\longrightarrow0.
\]
Since $t=O(d\log_2n)$, it follows that
$nt=o(m2^{-d})$ uniformly over the stated ranges.
Combining the probability identity
$\Pr_f(X\ge b)=2^{1-m}W((1-b)/2)$ from
\eqref{eq:weight-bias}, the weight bound
\eqref{eq:weight-reduction-input}, and
$W(2^{-\ell-1})\le2^{(1/8)m e^{-24\gamma\ell}}$
obtained above gives
\[
\Pr_f(X\ge2^{-\ell})
\le2^{1-m}\,2^{nt+(1-e^{-24\gamma\ell})m}
             \,2^{(1/8)m e^{-24\gamma\ell}}
=2^{1+nt-(7/8)m e^{-24\gamma\ell}}.
\]

Finally, the bounds $e^{-24\gamma\ell}\ge2^{-d}$ and
$m2^{-d}/n\to\infty$ from \eqref{eq:auxiliary-bias-bounds},
together with $nt=o(m2^{-d})$ proved above, give
\[
\frac{1+nt}{m e^{-24\gamma\ell}}
\le\frac1{m2^{-d}}+\frac{nt}{m2^{-d}}
\longrightarrow0
\]
uniformly over the stated ranges. Thus, for all sufficiently
large $n$, we have $1+nt\le(3/40)m e^{-24\gamma\ell}$ throughout
these ranges. Since $7/8-4/5=3/40$, it follows that
\[
\Pr_f(X\ge2^{-\ell})
\le2^{1+nt-(7/8)m e^{-24\gamma\ell}}
\le2^{-(4/5)m e^{-24\gamma\ell}},
\]
which is \eqref{eq:bias-tail}.

This completes the proof.
\end{proof}

\subsection{Full rank for a random evaluation matrix}

\begin{proof}[Proof of Lemma~\ref{lem:rank-retention}]
\emph{The weight bound and sampling probability.}
Put $\gamma=d/n\le\beta=1/2-\varepsilon$ and $m=m_{n,d}$.
Let $W_{n,d}(u)$ count the words of $\RM(n,d)$ of relative weight
at most $u$. Theorem~1.2 of Sberlo and Shpilka gives, in our notation,
\[
 \log_2W_{n,d}(2^{-\ell})
 \le Cn^4+17\left\{\frac{\ell}{1-\gamma}
       +\frac{2-\gamma}{(1-\gamma)^2}\right\}
       \gamma^{\ell-1}m
 \qquad(\ell\ge2),
\]
where $C$ is absolute. Since $\gamma<1/2$,
\[
 \frac1{1-\gamma}\le2,\qquad
 \frac{2-\gamma}{(1-\gamma)^2}
 =\frac1{1-\gamma}+\frac1{(1-\gamma)^2}\le6.
\]
Using also $\gamma^{\ell-1}\le\beta^{\ell-1}$, we obtain
\[
 \log_2W_{n,d}(2^{-\ell})
 \le Cn^4+17(2\ell+6)\beta^{\ell-1}m.
\]
Write $D=D_\varepsilon=544\ln2/\varepsilon$, so that
$D\ge8$, and take $q=Dm/N$.
The upper bound in \eqref{eq:binomial-entropy} gives
$m\le2^{nH_2(\beta)}$. Together with $D=544\ln2/\varepsilon$,
this gives the first inequality below.
Taylor's theorem gives $1-H_2(\beta)\ge2\varepsilon^2/\ln2$.
Using also $\varepsilon\ge n^{-1/3}$, we obtain
\[
 q\le\frac{544\ln2}{\varepsilon}\,2^{-n(1-H_2(\beta))}
 \le\frac{544\ln2}{\varepsilon}e^{-2n\varepsilon^2}
 \le544\ln2\,n^{1/3}e^{-2n^{1/3}}<1
\]
for all sufficiently large $n$, uniformly over the stated parameters.

\emph{Bounding the failure probability.}
The matrix $A_U$ fails to have full row rank precisely when
$A_U^{\mathrm T}$ has a non-zero kernel vector $c\in\F_2^m$.
Since
\[
 A_U^{\mathrm T}c=(f_c(x))_{x\in U},
\]
this is equivalent to a non-zero word of $\RM(n,d)$ vanishing at
every vertex in $U$. For a fixed word $f$ of relative weight $w$,
there are exactly $Nw$ vertices at which $f$ is non-zero. It vanishes
on $U$ precisely when all these vertices are excluded. Since each
is excluded independently with probability $1-q$, this has probability
\[
 (1-q)^{Nw}\le e^{-qNw}=e^{-Dmw}.
\]
The union bound therefore bounds the failure probability by
$\sum_{f\ne0}e^{-Dm\wt(f)}$. Words of small weight are more likely
to vanish on $U$, so we need to control how many such words there
are. We bound their contribution by grouping the weights into
intervals and applying the estimate for $W_{n,d}$ above.

Words of weight at least $1/4$ contribute at most $2^me^{-Dm/4}$.
The remaining non-zero words have weight at least $2^{-d}$ and are
covered by the intervals $[2^{-\ell-1},2^{-\ell})$,
$2\le\ell\le d$. There are at most $W_{n,d}(2^{-\ell})$ words
in the $\ell$th interval. Each has weight at least $2^{-\ell-1}$,
so contributes at most $e^{-Dm2^{-\ell-1}}$ to the union bound.
Thus the interval contributes at most
\[
 W_{n,d}(2^{-\ell})\,e^{-Dm2^{-\ell-1}}.
\]
Applying the weight estimate and using $e^{-x}=2^{-x/\ln2}$,
we bound this by
\[
\begin{aligned}
 &2^{Cn^4+17(2\ell+6)\beta^{\ell-1}m-Dm2^{-\ell-1}/\ln2}\\
 &\qquad=
 2^{Cn^4+m2^{-\ell}
 \left(34(2\ell+6)(2\beta)^{\ell-1}-D/(2\ln2)\right)}.
\end{aligned}
\]
For $0<2\beta<1$,
\[
 \ell(2\beta)^{\ell-1}
 \le\sum_{j=0}^{\ell-1}(2\beta)^j
 \le\frac1{1-2\beta}.
\]
Consequently, for every $\ell\ge2$,
\[
 (2\ell+6)(2\beta)^{\ell-1}
 \le\frac2{1-2\beta}+6
 \le\frac8{1-2\beta}=\frac4\varepsilon.
\]
Our choice of $D$ therefore gives
\[
 34(2\ell+6)(2\beta)^{\ell-1}
 \le\frac{136}{\varepsilon}=\frac{D}{4\ln2}.
\]
The contribution from the $\ell$th interval is therefore at most
\[
 2^{Cn^4-Dm2^{-\ell}/(4\ln2)}
 \le2^{Cn^4-Dm2^{-d}/(4\ln2)},
\]
where we used $\ell\le d$. There are at most $d$ such intervals.
Including the words of weight at least $1/4$ proves
\eqref{eq:retention-failure}.

If $m2^{-d}/n^4\to\infty$, both terms in
\eqref{eq:retention-failure} tend to zero. For the first, use
$D\ge8$ and $m\to\infty$. For the second, the negative term
in the exponent dominates $Cn^4$, while $d\le n$. This proves
the stated consequence for sequences.

\emph{The uniform assertion.} Fix $c>0$ and suppose $cn\le d$.
The lower bound in \eqref{eq:binomial-entropy} gives
\[
 m2^{-d}\ge\frac{2^{nH_2(d/n)-d}}{n+1}.
\]
By concavity, $H_2(u)\ge2u$ on $[0,1/2]$, so
\[
 m2^{-d}\ge\frac{2^d}{n+1}\ge\frac{2^{cn}}{n+1}.
\]
Using this bound, $D\ge8$, $d\le n$ and $m\ge n+1$ in
\eqref{eq:retention-failure}, we bound the failure probability by
\[
 e^{-(2-\ln2)(n+1)}
 +n\,2^{Cn^4-\frac{2}{\ln2}\frac{2^{cn}}{n+1}}.
\]
This tends to zero and depends only on $n$ and $c$. For each
$\zeta>0$, it is therefore at most $\zeta$ whenever
$n\ge n_1(c,\zeta)$, proving the final assertion.
\end{proof}

\section{Proof of the integral bounds}
\label{app:recovery-estimates}

We retain the notation of Subsection~\ref{sec:coordinate-recovery}.
In particular, $h(t)$ is the probability that a fixed column $A_v$
lies in the span over $\F_2$ of the other columns selected
independently with probability $t$.

\begin{proof}[Proof of Lemma~\ref{lem:recovery-bounds}]
We use the following form of the sharp-threshold inequality in
O'Donnell~\cite[Section~10.3, equation~(10.8)]{ODonnell2014}.
Select each of $K\ge2$ coordinates independently with probability
$t$, and let $u(t)$ be the probability of an event determined by
the selected set. Suppose that the event is increasing, meaning
that selecting additional coordinates cannot destroy it, and
non-constant, meaning that it occurs for some selected sets and
fails for others. Suppose also that a group of coordinate
permutations preserving the event can carry any coordinate to
any other. Then
\begin{equation}\label{eq:full-biased-threshold}
 u'(t)\ge\frac{c_0\ln K}{\sigma^2\ln(e/\sigma^2)}u(t)(1-u(t)),
 \qquad \sigma^2=4t(1-t),\quad 0<t<1,
\end{equation}
where $c_0>0$ is absolute. O'Donnell states the inequality using
a function taking values $-1$ and $1$. We encode a selected
coordinate by $-1$ and give the function value $-1$ when the event
occurs. Its variance is then $4u(t)(1-u(t))$, and the factor $4$
is absorbed into $c_0$.

We apply \eqref{eq:full-biased-threshold} to the event
\[
 A_v\in\Span_{\F_2}\{A_u:u\in U_t^{(v)}\}.
\]
The $K=N-1$ coordinates specify which vertices other than $v$
are selected, and the event probability is $h(t)$. Selecting more
columns cannot destroy the event, so it is increasing.

By the affine invariance established in
Subsection~\ref{sec:coordinate-recovery}, affine maps over $\F_2$
fixing $v$ preserve this event. They act transitively on
$V\setminus\{v\}$: given $w,z\ne v$, choose an invertible linear
map $L$ over $\F_2$ with $L(w-v)=z-v$, by extending these
non-zero vectors to bases. Then $x\mapsto v+L(x-v)$ fixes $v$
and sends $w$ to $z$.

The event is non-constant. The constant row of $A$ gives $A_v\ne0$,
so it fails when no other column is selected. A row indexed by
$J\subseteq[n]$ contains $2^{n-|J|}$ ones, which is even because
$|J|\le d<n$. Hence, over $\F_2$,
\[
 \sum_{x\in V}A_x=0,
 \qquad A_v=\sum_{x\ne v}A_x.
\]
The event therefore holds when all other columns are selected.
Thus $h(0)=0$ and $h(1)=1$. Both selections have positive
probability when $0<t<1$, so $0<h(t)<1$ throughout this interval.

For every $0<t<1$, we have
\[
 \sigma^2\ln(e/\sigma^2)\le4t\ln(e/t).
\]
Indeed, when $t\le1/2$, the inequalities $2t\le\sigma^2\le4t$
give $\ln(e/\sigma^2)\le\ln(e/t)$, and multiplying the bounds
for the two positive factors gives the result. When $t\ge1/2$,
we instead use
\[
 \sigma^2\ln(e/\sigma^2)\le1\le4t\ln(e/t),
\]
since $x\ln(e/x)\le1$ for $0<x\le1$. Also, since $N=2^n$,
\[
 \ln(N-1)\ge(n-1)\ln2\ge\frac{\ln2}{2}\,n
 \qquad(n\ge2).
\]
Substitution in \eqref{eq:full-biased-threshold} therefore gives
\[
 \frac{d}{dt}\ln\frac{h(t)}{1-h(t)}
 =\frac{h'(t)}{h(t)(1-h(t))}
 \ge\frac{\kappa}{t\ln(e/t)},
\]
with $\kappa=c_2n$ for an absolute $c_2>0$.
This proves~(i).

To prove~(ii), we first bound $h(t)$ for $0<t\le p$.
Recall that $b_p=\ln(e/p)$. Integrating
\eqref{eq:uniform-threshold} from $t$ to $p$ gives
\[
 \ln\frac{h(p)}{1-h(p)}-\ln\frac{h(t)}{1-h(t)}
 \ge\kappa\int_t^p\frac{ds}{s\ln(e/s)}
 =\kappa\ln\frac{\ln(e/t)}{b_p}.
\]
Rearranging and exponentiating, we obtain
\[
 h(t)\le\frac{h(t)}{1-h(t)}
 \le\frac{h(p)}{1-h(p)}
       \left(\frac{b_p}{\ln(e/t)}\right)^\kappa.
\]
We now integrate over $[0,p]$. Substituting $t=pe^{-u}$,
so that $\ln(e/t)=b_p+u$, gives
\[
 \begin{aligned}
 \int_0^p h(t)\,dt
 &\le\frac{h(p)}{1-h(p)}\,p
       \int_0^\infty e^{-u}(1+u/b_p)^{-\kappa}\,du\\
 &\le\frac{h(p)}{1-h(p)}\,p
       \int_0^\infty(1+u/b_p)^{-\kappa}\,du
  =\frac{h(p)}{1-h(p)}\frac{p b_p}{\kappa-1},
 \end{aligned}
\]
provided $\kappa>1$. If $h(p)\le1/2$, the factor
$h(p)/(1-h(p))$ is at most one, proving
\eqref{eq:recovery-left}.

To prove~(iii), we want to bound $\int_p^1(1-h(t))\,dt$.
By the third identity in \eqref{eq:central-area-rank}, multiplying
this integral by $N$ gives $m-F(p)$, from which we will bound the
expected nullity. We first integrate
\eqref{eq:uniform-threshold} to bound $1-h(t)$ for each $t\ge p$.
We then integrate that upper bound over $[p,1]$.

Integrating \eqref{eq:uniform-threshold} from $p$ to $t$,
for $p\le t<1$, gives
\[
 \ln\frac{h(t)}{1-h(t)}-\ln\frac{h(p)}{1-h(p)}
 \ge\kappa\int_p^t\frac{ds}{s\ln(e/s)}
 =\kappa\ln\frac{\ln(e/p)}{\ln(e/t)}.
\]
Exponentiating and taking reciprocals yields
\[
 \frac{1-h(t)}{h(t)}
 \le\frac{1-h(p)}{h(p)}
       \left(\frac{\ln(e/t)}{\ln(e/p)}\right)^\kappa
 =z_p\left(\frac{\ln(e/t)}{b_p}\right)^\kappa.
\]
Since $h(t)\le1$, the same expression bounds $1-h(t)$.
We bound the logarithmic factor by a power of $p/t$, whose
integral can be calculated explicitly. For $p\le t<1$,
\[
 \frac{\ln(e/t)}{b_p}=1-\frac{\ln(t/p)}{b_p},
 \qquad 0\le\frac{\ln(t/p)}{b_p}<1.
\]
Applying $1-x\le e^{-x}$ with $x=\ln(t/p)/b_p$ gives
\[
 \left(\frac{\ln(e/t)}{b_p}\right)^\kappa
 \le\exp\left(-\frac{\kappa}{b_p}\ln(t/p)\right)
 =\left(\frac pt\right)^{\kappa/b_p}.
\]
Consequently,
\[
 1-h(t)\le z_p\left(\frac pt\right)^{\kappa/b_p}
 \qquad(p\le t<1).
\]
With $p$ fixed, the upper bound is a constant multiple of
$t^{-\kappa/b_p}$.

We can now bound the required integral. If $b_p<\kappa$, the
exponent $\kappa/b_p$ is greater than $1$.
Integrating over $[p,1]$ and then enlarging the integral of the
positive power function to $[p,\infty)$, we obtain
\[
 \int_p^1(1-h(t))\,dt
 \le z_p\int_p^\infty(p/t)^{\kappa/b_p}\,dt
 =\frac{z_ppb_p}{\kappa-b_p}.
\]
This proves \eqref{eq:recovery-tail}. Multiplying by $N$ and
using the third identity in \eqref{eq:central-area-rank} gives
\eqref{eq:recovery-rank-bound}.
\end{proof}

\section{Proof of Theorem~\ref{thm:near-central}(ii)}
\label{app:near-central-ii}

We first prove the expected-rank bound in
Lemma~\ref{lem:retention-excess}, then compare the two entropy
terms in \eqref{eq:entropy-transfer} and choose the sampling
probability.

\begin{proof}[Proof of Lemma~\ref{lem:retention-excess}]
Recall that $\alpha=m/N$ and $a=\ln(1/\alpha)$.
We select each column independently with probability
$q=(1+\eta)\alpha$, giving $(1+\eta)m$ columns on average.

Recall that $h(t)$ is the probability that a fixed column lies in
the span of the other columns selected with probability $t$.
The function $h$ is continuous with $h(0)=0$ and $h(1)=1$, and
Lemma~\ref{lem:recovery-bounds}(i) makes it strictly increasing on $(0,1)$.
Thus there is a unique $\theta\in(0,1)$ with $h(\theta)=1/2$.
We first show that
$\theta$ is close to $\alpha$.

Lemma~\ref{lem:central-nullity} gives
\[
 \int_0^\alpha h(t)\,dt
 =\int_\alpha^1(1-h(t))\,dt
 \le C_I\alpha\frac an.
\]
If $\theta<\alpha$, then $h(t)\ge1/2$ on $[\theta,\alpha]$.
If $\theta>\alpha$, then $1-h(t)\ge1/2$ on $[\alpha,\theta]$.
In either case, the corresponding integral is at least
$|\theta-\alpha|/2$. Hence
\[
 |\theta-\alpha|\le2C_I\alpha\frac an.
\]
Choose $C_0$ to satisfy $C_0\ge4C_I$ and
$C_0\ge33/(5c_2)$. The hypothesis $\eta\ge C_0a/n$
then gives $2C_Ia/n\le\eta/2$. Since $q-\alpha=\alpha\eta$,
\[
 q-\theta=(q-\alpha)-(\theta-\alpha)
 \ge\alpha\eta-2C_I\alpha\frac an
 \ge\frac{\alpha\eta}{2}.
\]
The hypothesis $a/n\le\delta_0$ also gives
$\theta\ge\alpha(1-2C_I\delta_0)$. Reducing $\delta_0$ so that
$2C_I\delta_0\le1/10$ ensures $\theta\ge0.9\alpha$.
Since $\eta\le1/10$, we have $q\le1.1\alpha$ as well.
Thus $[\theta,q]\subseteq[0.9\alpha,1.1\alpha]$.
Throughout this interval,
\[
 t\ln(e/t)
 \le1.1\alpha\{a+1+\ln(10/9)\}
 \le4\alpha a,
\]
where the last inequality uses $a\ge\ln2$.

Integrating \eqref{eq:uniform-threshold} from $\theta$ to $q$,
with $\kappa=c_2n$ and $h(\theta)=1/2$, gives
\[
 \ln\frac{h(q)}{1-h(q)}
 \ge\frac{c_2n(q-\theta)}{4\alpha a}
 \ge\frac{c_2}{8}\frac{n\eta}{a}.
\]
Put $\mu=c_2/8$. Exponentiating and taking reciprocals gives
\[
 \frac{1-h(q)}{h(q)}\le e^{-\mu n\eta/a}.
\]

To pass from this estimate to the expected rank, apply
\eqref{eq:recovery-rank-bound} with $p=q$. Its condition
$b_q<\kappa$ holds after reducing $\delta_0$ to at most $c_2/6$,
since $q\ge\alpha$, $a\ge\ln2$, and $a\le\delta_0n$ give
\[
 b_q=\ln(e/q)\le a+1\le3a\le\frac{c_2n}{2}<\kappa.
\]
Using $m-F(q)=m\rho$, $m=N\alpha$, and $q/\alpha\le1.1$, we obtain
\[
 \rho\le\frac{Nz_q q b_q}{m(\kappa-b_q)}
 =\frac{1-h(q)}{h(q)}
   \frac{q\ln(e/q)}{\alpha\{c_2n-\ln(e/q)\}}.
\]
Since $c_2n-\ln(e/q)\ge c_2n/2$, it follows that
\[
 \rho\le\frac{11}{10}\frac{3a}{c_2n/2}
                    e^{-\mu n\eta/a}
 =\frac{33}{5c_2}\frac an e^{-\mu n\eta/a}
 \le C_0\frac an e^{-\mu n\eta/a}.
\]
Since the exponential factor is at most $1$ and
$C_0a/n\le\eta$, this also gives $\rho\le\eta$.
\end{proof}

\begin{proof}[Proof of Theorem~\ref{thm:near-central}(ii)]
We first obtain an error bound whenever $0<\eta\le1/10$ and
$0\le\rho\le\eta$. Lemma~\ref{lem:retention-excess} supplies the
latter condition for the choice made below. At
$q=(1+\eta)\alpha$, the expected nullity is $m(\eta+\rho)$.
Combining \eqref{eq:entropy-transfer} with
$\log_2Q_{n,d}\le NH_2(\alpha)+1$ from
Lemma~\ref{lem:cover-estimates}(iii) gives
\[
 \begin{aligned}
 \log_2Q_{n,d}-\log_2T_{n,d}
 &\le N\{H_2(\alpha)-H_2(q)+H_2(\alpha(\eta+\rho))\}\\
 &\quad+O\!\left(\frac{mn^2}{2^d}+n\right),
 \end{aligned}
\]
where the additional $1$ has been included in the $O$-term.
We need to bound these entropy terms together. The leading terms
arising from the additional $\eta m$ columns will cancel.

By concavity of $H_2$ and the identity $q-\alpha=\eta\alpha$,
\[
 \begin{aligned}
 H_2(q)-H_2(\alpha)
 &\ge(q-\alpha)H_2'(q)\\
 &=\frac{\eta\alpha}{\ln2}
       \{a-\ln(1+\eta)+\ln(1-q)\}\\
 &\ge\frac{\eta\alpha a}{\ln2}-2\alpha\eta.
 \end{aligned}
\]
Here $\eta\le1/10$ and $q=(1+\eta)\alpha\le0.55$, since
$\alpha\le1/2$. Hence
\[
 \ln(1+\eta)-\ln(1-q)
 \le\ln(22/9)<2\ln2.
\]
On the other hand, the bound $H_2(x)\le x\ln(e/x)/\ln2$ gives
\[
 \begin{aligned}
 H_2(\alpha(\eta+\rho))
 &\le\frac{\alpha(\eta+\rho)}{\ln2}
       \{a+\ln(1/(\eta+\rho))+1\}\\
 &\le\frac{\alpha(\eta+\rho)a}{\ln2}
       +\frac{2\alpha\eta}{\ln2}\{\ln(1/\eta)+1\}\\
 &\le\frac{\alpha(\eta+\rho)a}{\ln2}
       +\frac4{\ln2}\alpha\eta\ln(1/\eta).
 \end{aligned}
\]
Here we kept the term involving $a$ separate and used
$\eta\le\eta+\rho\le2\eta$ for the remaining term.
The last inequality follows from $\eta\le1/10$, which gives
$\ln(1/\eta)\ge1$.
Subtracting the lower bound for $H_2(q)-H_2(\alpha)$ from this
upper bound cancels $\eta\alpha a/\ln2$. It gives
\[
 \begin{aligned}
 H_2(\alpha)-H_2(q)+H_2(\alpha(\eta+\rho))
 &\le\frac{\alpha\rho a}{\ln2}
       +\left(\frac4{\ln2}+2\right)\alpha\eta\ln(1/\eta)\\
 &\le K\alpha\{\rho a+\eta\ln(1/\eta)\},
 \end{aligned}
\]
where $K=4/\ln2+2$ is fixed throughout this proof.

Lemma~\ref{lem:cover-estimates}(iii) and $a\ge\ln2$ give
$\log_2Q_{n,d}\ge\Lambda/2\ge c m(a+1)/2$, where $c>0$ is
absolute, and $\Lambda=NH_2(\alpha)$. Dividing by this lower bound,
using $N\alpha=m$ and
$a/(a+1)\le1$, and applying the normalised deletion estimate from
Subsection~\ref{sec:central-common-estimates}, we obtain
\begin{equation}\label{eq:entropy-cancellation}
 E_{n,d}\le\frac{2K}{c}\left\{\rho+
       \frac{\eta\ln(1/\eta)}{a+1}\right\}+O(e^{-c'n}).
\end{equation}

We now choose $\eta$ to make $\rho=O(1/n)$. The bound in
Lemma~\ref{lem:retention-excess} has prefactor $C_0a/n$, so
making its exponential factor at most $1/(a+2)$ will suffice.
Set
\[
 \eta=B\frac an\ln(a+2),
\]
where $B$ is a fixed absolute constant satisfying
$B\ln2\ge\max\{C_0,1\}$ and $\mu B\ge1$.
Since $\ln(a+2)\ge\ln2$, these choices ensure
$\eta\ge C_0a/n$ and $\eta\ge a/n$.
The theorem assumes $(a/n)\ln(a+2)\le\delta_1$.
Choose
\[
 \delta_1\le\min\{1/(10B),\delta_0\ln2\}.
\]
Then $\eta\le B\delta_1\le1/10$ and
$a/n\le\delta_1/\ln(a+2)\le\delta_0$.
Lemma~\ref{lem:retention-excess} therefore applies. Our choice
of $\eta$ gives $e^{-\mu n\eta/a}=(a+2)^{-\mu B}$, and hence
\[
 \rho\le C_0\frac{a}{n(a+2)^{\mu B}}
 \le\frac{C_0}{n}\frac{a}{a+2}\le\frac{C_0}{n}.
\]
Also, $\eta\ge a/n$ implies $\ln(1/\eta)\le\ln(n/a)$, while
\[
 \frac{\eta}{a+1}
 =\frac{B}{n}\frac{a}{a+1}\ln(a+2)
 \le\frac{B\ln(a+2)}n.
\]
Substituting these bounds into \eqref{eq:entropy-cancellation}
gives
\[
 E_{n,d}\le\frac{2K}{c n}
       \{C_0+B\ln(a+2)\ln(n/a)\}+O(e^{-c'n}).
\]
The exponentially small term is at most $1/n$ in absolute value
for large $n$. Taking $C\ge\max\{2KC_0/c+1,2KB/c\}$, and large
enough also for part~(i), gives
\[
 E_{n,d}\le\frac{C\{1+\ln(a+2)\ln(n/a)\}}n.
\]
\end{proof}

\section*{Declaration of generative AI use}

Anthropic's Claude and OpenAI's ChatGPT were used to assist
with literature research, mathematical checking, and revisions
of the proofs and exposition.
The author takes full responsibility for the contents of the paper.

\end{document}